\documentclass[11pt]{amsart}

\usepackage{latexsym,amsfonts,amsmath,amssymb,amsthm}
\usepackage{multirow}
\usepackage[usenames,dvipsnames]{color}
\usepackage[colorlinks=true,linkcolor=cyan,citecolor=DarkOrchid, urlcolor=BurntOrange, pagebackref]{hyperref}
\usepackage{url}
\usepackage[top=1in, bottom=1in, left=1in, right=1in]{geometry}

\usepackage{graphicx}
\usepackage{subfigure}
\usepackage{float}
\usepackage[misc]{ifsym}

\newcommand{\bd}{\boldsymbol}

\newtheorem{theorem}{Theorem}

\newtheorem{lemma}{Lemma}[section]
\newtheorem{prop}[lemma]{Proposition}

\theoremstyle{definition}

\newtheorem{remark}[lemma]{Remark}
\newtheorem{assumption}[lemma]{Assumption}

\let\llldots=\ldots
\def\ldots{\llldots{}}

\newcommand{\beqn}{\begin{equation}}
\newcommand{\eeqn}{\end{equation}}

\newcommand{\tr}{\text{tr}}

\newcounter{smalllist}

\let\llldots=\ldots
\def\ldots{\llldots{}}

\numberwithin{equation}{section}

\begin{document}
\title[Evolution of scattering data under FPU lattices]{On the evolution of scattering data under perturbations of the Toda lattice}
\author[D.~Bilman]{Deniz Bilman}
\address{Deniz Bilman\\
         Department of Mathematics, Statistics and Computer Science\\ 
         University of Illinois at Chicago\\
         851 S. Morgan Street\\
         Chicago, IL\\
         }
\email{dbilma2@uic.edu}
\author[I.~Nenciu]{Irina Nenciu}
\address{Irina Nenciu\\
         Department of Mathematics, Statistics and Computer Science\\ 
         University of Illinois at Chicago\\
         851 S. Morgan Street\\
         Chicago, IL \textit{and} Institute of Mathematics ``Simion Stoilow''
     of the Romanian Academy\\ 21, Calea Grivi\c tei\\010702-Bucharest, Sector 1\\Romania\\
     }
\email{nenciu@uic.edu}

\begin{abstract}
We present the results of an analytical and numerical study of the long-time behavior for certain 
Fermi-Pasta-Ulam (FPU) lattices viewed as perturbations of the completely integrable Toda lattice. 
Our main tools are the direct and inverse scattering transforms
for doubly-infinite Jacobi matrices, which are well-known to linearize the Toda flow. 
We focus in particular on the evolution of the associated scattering data under the perturbed vs.~the unperturbed equations. We find that the eigenvalues present initially in the scattering data converge to new, slightly perturbed eigenvalues under the perturbed dynamics of the lattice equation. To these eigenvalues correspond solitary waves that emerge from the solitons in the initial data. We also find that new eigenvalues emerge from the continuous spectrum as the lattice system is let to evolve under the perturbed dynamics.
\end{abstract}


\maketitle

\section{Introduction}\label{S:1}

The purpose of this work is to numerically investigate the long time behavior of solutions to certain
perturbations of the completely integrable Toda lattice. In its wider context, this problem lies within
a set of fundamental questions concerning the qualitative features of solutions to nonlinear 
Hamiltonian partial differential or difference equations (PDEs or P$\Delta$Es,
respectively). It is well known that in these types of equations 
several different phenomena can appear over long times, among them blow-up, scattering to the free 
evolution, or the emergence of stable nonlinear structures, such as solitary waves and breather solutions.
A detailed description of any such phenomenon often depends on the precise structure of the 
equation in question. However, there is the common belief, usually referred to as the soliton resolution conjecture, 
that (in the absence of finite time blow-up) generic solutions can be decomposed at large times into 
a sum of solitary waves plus a dispersive tail (i.e., radiation). 

Loosely speaking, what we will call here solitary waves are special solutions to nonlinear PDEs or P$\Delta$Es
which travel with constant speed and without changing their profile. Their existence reflects a 
certain balance between nonlinear and dispersive effects in a given evolution equation.
The statement of the soliton resolution conjecture seems 
to be the natural scenario for the long time behavior of Hamiltonian PDEs or P$\Delta$Es, 
but there are at this time relatively few rigorous results in this direction. Part of the reason for this is that 
many current mathematical techniques, 
while extremely powerful, are intrinsically linear. But the very nature of the problem means that we are 
interested in studying the competition between nonlinearity and dispersion over infinite time intervals. 

There is one very well-known example, however, in which the description of the long
time asymptotic behavior of solutions is fully known, namely completely integrable equations. These are Hamiltonian equations which have so ``many''
conserved quantities that they can be ``completely'' integrated. In the infinite dimensional setting,
the definition of complete integrability is not fully set. We will adopt the (highly practical) point of view that a
partial differential or difference equation is completely integrable if it can be linearized through some
bijective transformation. The transformation of choice is provided by
the direct and inverse scattering data associated to the Lax operator of the integrable equation, and it is
this scattering transform which can be thought of as a non-linear version of the Fourier transform. One
of the advantages of the scattering transform is that it separates, in a very precise, quantitative sense
the ``parts'' of the initial data (namely each eigenvalue and associated norming constant) 
which lead to the asymptotic emergence of each soliton, and those (mainly the reflection coefficient) that lead to 
the dispersive tail. This understanding, combined with the nonlinear stationary phase/steepest descent techniques 
introduced by P. Deift and X. Zhou \cite{DZ_RHP}, has yielded a large number of 
rigorous asymptotic results for various completely integrable PDEs (see, for example, \cite{DVZ} and \cite{DZ_PII},
among many others).

The question we wish to investigate in this paper is whether or not this type of analysis, 
based on the study of the scattering data,
can be extended to apply to certain (non-integrable) perturbations of completely integrable PDEs or P$\Delta$Es. 
We are aiming at a more quantitative understanding of the behavior of solutions than
is offered by (infinite dimensional) KAM theory. At the same time, we want to work in a situation in which
the perturbations are interesting in and of themselves, but which also allow us to minimize 
the technical challenges (both analytic and numeric) which
are not directly related to the question of the evolution of
solitary waves under perturbations of the integrable case.

We focus on certain so-called Fermi-Past-Ulam (FPU) lattices (named after the study in~\cite{FPU}), viewed as perturbations of the 
(completely integrable) Toda lattice. There already exists an
extensive body of work, both numerical and analytical, focused on the study of
the Toda and FPU lattices, and it is unfortunately impossible to include a 
comprehensive bibliography here. We wish however to mention the work
of G. Friesecke, R.~L. Pego and A.~D. Wattis (see \cite{FP1,FP2,FP3,FP4,FW})
as particularly relevant to our study, since they proved the existence of solitary waves
for general FPU lattices.  Furthermore, we rely heavily on the connection
between the Toda lattice and Jacobi matrices, in particular on the scattering
theory for Jacobi matrices (see, for example, \cite{T} and \cite{Tesch01} 
and the references therein). Finally, we should mention that the very same perturbed
evolution equations we consider (see \eqref{E:eom_pert_pq} and \eqref{E:eom_p_L})
are also studied in \cite{IST}. However, there the authors consider the question
of Lieb-Robinson bounds, which provide information on the behavior of solutions for short
(vs. arbitrarily long) times.

Numerical study of the spectrum of a linear problem associated with an integrable equation is a method that has been previously used to understand nearly-integrable dynamics. In a series of papers \cite{AHS3, AHS4, AHS2, AHS1}, following \cite{AH} and \cite{MBFO}, M.~Ablowitz, B.~Herbst, and C.~Schober studied the effects of the perturbations that are induced by truncations, discretizations, and roundoff errors present in the numerical computations on solutions of certain integrable systems, such as the nonlinear Schr\"{o}dinger and the sine-Gordon equations on periodic spatial domain. The authors numerically tracked the evolution of eigenvalues and the phase portraits associated with initial data in the neighborhood of homoclinic manifolds, and revealed how numerically induced perturbations can lead to chaos due to the unstable nature of the initial data. In our work, however, we consider certain fixed perturbations of the integrable Toda Hamiltonian and investigate long-time behavior of scattering data associated with solutions that evolve under the dynamics that result from the perturbed Hamiltonian. In an earlier study \cite{FFM}, W.~Ferguson, H.~Flaschka, and D.~McLaughlin numerically investigated the relationship between excitation of action-variables and corresponding solutions of the periodic Toda lattice for a variety of initial conditions. The authors also computed the Toda action variables corresponding to solutions of the original FPU lattice, which can be treated as a perturbation of the integrable Toda lattice, and observed that the action variables remain nearly constant under the non-integrable FPU dynamics. Our study, as mentioned above, focuses on the long-time behavior of scattering data and solutions for perturbations of the doubly-infinite Toda lattice.

One can study the behavior of solutions to
(not necessarily integrable) PDEs and P$\Delta$Es in various other asymptotic regimes,
(among them, semi-classical and continuum limits), or focus on questions of stability of certain special solutions; 
these are very well-established and well-studied problems, using a variety of techniques, including some
from the completely integrable arsenal. As an example,
let us mention the recent and on-going work of B. Dubrovin, T. Claeys, T. Grava, C. Klein and K. McLaughlin, 
(see, for example, \cite{CG_09,CG_10,DGK}) who, inspired by ideas and conjectures of B. Dubrovin \cite{D}, 
study the behavior of a large class of Hamiltonian perturbations of Burgers' equation which are required to be 
integrable up to a certain order in the perturbation parameter. 

Our paper is organized in two separate parts. In the first part, comprising Section~\ref{S:Background}, 
we present the analytical background of our problem. This background is supplemented by several results concerning
the evolution equations under consideration, which can be found in the Appendix. 
These results consist of global-in-time existence and (spatial) decay properties 
of the solutions on the one hand, as well as set-up and properties of the associated scattering data. In particular, we
deduce the evolution equations for the scattering data -- which, unsurprisingly, turn out to be nonlinear, nonlocal perturbations
of the linear evolution equations that the scattering follow in the integrable, Toda case. 
The study of the long-time asymptotics of the perturbed equations for the scattering data are known to be highly 
nontrivial even in the case of dispersive equations (see, for example, the seminal paper \cite{DZ_Pert}). 
Ultimately, we wish to perform the same type of analysis as in \cite{DZ_Pert} for the equations considered here, but the situation 
is even more complicated due to the existence, in this case, of solitary, traveling waves.

Thus the numerical part of our project is an essential first step in the study of the long-time asymptotics for solutions
of these perturbations of the Toda lattice. In the second, numerical, part of the paper we analyze and report on the long-time
behavior of solutions and scattering data for a variety of initial data and of perturbations. Even more results can be found
on the project's webpage, \texttt{http://bilman.github.io/toda-perturbations}. In Section~\ref{S:scheme}
we present our numerical scheme and justify its validity in approximating the actual solutions of the perturbed equations. 
Section~\ref{S:soliton_initial} contains our results for initial data given by a Toda soliton, which is then allowed to evolve under the 
perturbed equation. We study not only the shape of the solution, but also, more importantly, the behavior of the scattering data,
with a focus on the eigenvalues. In Section~\ref{S:soliton_free_data}, we perform the same analysis for initial data which are
without solitons in the Toda lattice -- or, equivalently, for initial data whose scattering data do not have any eigenvalues. Finally,
Section~\ref{S:clean} contains results on so-called clean solitary waves -- that is, numerical solutions which correspond to
the analytical solitary waves whose existence is known for these perturbed lattices. For these clean solitary waves
we compute the scattering data in order to compare them with the long-time asymptotics obtained from non-cleaned
initial data, as well as perform numerical experiments to simulate collisions and multi-soliton type solutions.

What we find consistently in our study is that, while the eigenvalues present at time $0$ in the scattering data of a solution
are, naturally, no longer constant, they do converge very fast to new, slightly perturbed
values. This is consistent with what was reported in \cite{FFM}. However, in addition to this, we also observe the emergence
of new eigenvalues from the edges of the absolutely continuous (ac) spectrum. The behavior of these new eigenvalues
is more difficult to analyze than that of the persistent, initial eigenvalues, as they do not appear to stabilize
at values outside of the ac spectrum. On the physical side, that means that solutions remember, as time goes to infinity, the solitary 
waves ``contained'' in their initial data, albeit with slightly modified velocities and amplitudes, 
but that in addition new, small ``bumps"
appear instantaneously when the perturbed evolution starts. These new bumps do not appear to approach stationary
waves, but rather their speeds and amplitudes seem to decrease with time, which could make them in fact part of
the dispersive tails of the solutions. However, our current numerical analysis cannot confidently predict the long time
behavior of these new waves (or, equivalently, of their associated eigenvalues), and a different project,
based on the numerical study of Riemann-Hilbert problems, is currently investigating this phenomenon.

\section{Background and Evolution Equations}\label{S:Background}
We consider the classical problem of a 1-dimensional chain
of particles with nearest neighbor interactions. We will assume throughout that the
system is uniform (contains no impurities) and that the mass of each
particle is normalized to 1. The equation that governs the evolution is then:
\begin{equation}\label{E:FPUeq}
\partial_t^2{q}_n(t)=V_\varepsilon'\big(q_{n+1}(t)-q_n(t)\big)-V_\varepsilon'\big(q_n(t)-q_{n-1}(t)\big)\,,
\end{equation}
where $q_n(t)$ denotes the displacement of the $n^{\text{th}}$
particle from its equilibrium position, and $V$ is the interaction potential between neighboring
particles. We will focus on potentials of the form:
\begin{equation*}
V_\varepsilon(r) = V_0(r) + \varepsilon u(r)\,,
\end{equation*}
where $V_0$ is the well-known Toda potential (see, for example, \cite{Toda}) $V_0(r)=e^{-r}+r-1$,
$\varepsilon >0$, and $u$ is a (well-chosen) perturbation potential.
In order to ensure that we are studying a meaningful case, 
we must make certain assumptions:
\begin{assumption}\label{A}
We consider $V_\varepsilon\in C^2(\mathbb{R})$ such that:
\begin{enumerate}
\item[{\rm (}i{\rm )}] $V_\varepsilon(r)\geq 0$ for all $r\in\mathbb R$;
\item[{\rm (}ii{\rm )}] $V_\varepsilon(0)=V_\varepsilon'(0)=0$ and $V_\varepsilon''(0)>0$;
\item[{\rm (}iii{\rm )}] $V_\varepsilon(r)\to\infty$ as $|r|\to\infty$;
\item[{\rm (}iv{\rm )}] $V_\varepsilon$ is super-quadratic on at least one side of the origin, i.e.,
\begin{equation*}
\frac{V_\varepsilon(r)}{r^2}\text{ increases strictly with } |r|\text{ for all } r\in\Lambda \text{, with } \Lambda=\mathbb R_+ \text{ or }
\Lambda=\mathbb R_-\,.
\end{equation*}
\end{enumerate}
\end{assumption}

It is then known (see, for example, \cite{IST}) that, if $V_\varepsilon$ satisfies Assumption \ref{A}, and
if we rewrite \eqref{E:FPUeq} as a first order system
\begin{equation}\label{E:eom_ps}
\begin{cases}
\begin{aligned}
\partial_t p_n &= V'_{\varepsilon}(q_{n+1}-q_n) - V'_{\varepsilon}(q_n-q_{n-1}) \\
\partial_t q_n &= p_n\,,
\end{aligned}
\end{cases}
\end{equation}
then the system is well-posed globally in time in $\big((p_n)_{n\in\mathbb Z},(q_{n+1}-q_n)_{n\in\mathbb Z}\big)
\in\ell^2(\mathbb{Z}) \oplus \ell^2(\mathbb{Z})$ for any initial data. The system is Hamiltonian, with Hamiltonian function
\begin{equation*}
\mathcal{H}_{\varepsilon}(p,q) = \sum_{n\in\mathbb{Z}}\left[\frac{p_n^2}{2} + V_{\varepsilon}(q_{n+1}-q_n)\right]\,,
\end{equation*}
which is finite and conserved in time along solutions. Furthermore, the solutions are bounded uniformly in time:
\begin{equation*}
\big\|\big(p_n(t)\big)_{n\in\mathbb Z}\big\|_{\ell^2(\mathbb Z)\cap\ell^\infty(\mathbb Z)}+
\big\|\big(q_{n+1}(t)-q_n(t)\big)_{n\in\mathbb Z}\big\|_{\ell^2(\mathbb Z)\cap\ell^\infty(\mathbb Z)}\leq C\,,
\end{equation*}
for some constant $C>0$ which depends only on the initial data.

The Toda potential $V_0$ is easily seen to satisfy Assumption~\ref{A}.
We note that, in order to ensure that $V_\varepsilon$ satisfies Assumption~\ref{A}, 
$u$ must satisfy $u(0)=u'(0)=0$, but is not required to be non-negative or growing at $\pm\infty$;
in fact, $u$ is allowed to be a bounded function since $V_0$ already satisfies the necessary growth conditions. 
In light of this discussion, the perturbations we consider 
in our numerical studies are not limited to convex, non-negative functions such as $u(r)=r^{2N}$, for $N\in\mathbb{N}$ with 
$N\geq 2$; we also include functions which are not 
convex or positive, such as $u(r)=r^{2N+1}$, for $N\in\mathbb{N}$ (suitably scaled by $\varepsilon$), 
as well as bounded functions such as $u(r) = 1 - \cos{r}$.

In addition to the global in time well-posedness, we know from Theorem 1 of \cite{FW} that
if $V_\varepsilon$ satisfies Assumption~\ref{A}, then there exists a constant $K_0 \geq 0$ such that for every $K>K_0$ the 
system \eqref{E:eom_ps} possesses a nontrivial traveling solitary wave with finite kinetic energy 
and with average potential energy $K$. 
These solitary waves $(q_n)_{n\in\mathbb{Z}}$ have the following properties:
\begin{itemize}
\item they are monotone functions, increasing (expansion waves) if $\Lambda = \mathbb{R}^{+}$ and decreasing (compression waves) if $\Lambda=\mathbb{R}^{-}$.
\item they are localized, in the sense that $q_{n+1} - q_n \rightarrow 0$ as $n\rightarrow \pm\infty$.
\item they are supersonic; that is, their wave speeds $c$ satisfy $c^2 > V''(0)$.
\end{itemize}

In the case of the Toda lattice ($\varepsilon=0$), these waves are solitons, and they are known
explicitly. For example, 1-soliton solutions for the displacements are given by the 2-parameter family
\begin{equation}\label{E:1sol_q}
q_n(t) = \log{\left(\frac{1+\frac{\gamma}{1-e^{-2k}}e^{-2k(n-1) - 2k\sigma c t}}{1+\frac{\gamma}{1-e^{-2k}}e^{-2kn - 2k\sigma c t}}\right)}\,,
\end{equation}
where $\gamma>0$, $k > 0$ is the wave number, $c = \frac{\sinh{k}}{k} > 1$ is the speed of propagation, 
and $\sigma = \pm 1$ is the constant determining the direction of propagation. 

As is now well-known, the existence of solitons
in this case is due to the complete integrability of the Toda lattice. 
Complete integrability of the Toda lattice was proven by H.~Flaschka in 1974
by introducing a change of variables that allowed him to set
the system in Lax pair form \cite{Fla1,Fla2}.
Since our aim is to follow the behavior under
the perturbed evolution of quantities linked to the completely integrable lattice, we will start with
Flaschka's change of variables, but apply it to the perturbed lattice:
\begin{equation*}
a_n=\frac12 e^{-\frac{q_{n+1}-q_n}{2}} \quad\text{and}\quad
    b_n=-\frac12 p_n\,. 
\end{equation*}
Note that the new variables satisfy $(a_n)_{n\in\mathbb{Z}},(1/a_n)_{n\in\mathbb Z},(b_n)_{n\in\mathbb{Z}}
\in\ell^\infty(\mathbb{Z})$. We remark here that for any potential $V_\varepsilon$ considered in this study, the super quadratic 
growth condition (Assumption~\ref{A}(iv)) is satisfied on $\Lambda=(-\infty, 0)$, and therefore solitary waves obey $a_n > \frac{1}{2}$ for all $n\in\mathbb{Z}$.
Introduce the second-order linear difference operators $L$ and $P$ defined on $\ell^2(\mathbb{Z})$ by
\begin{align*}
(Lf)_n &= a_{n-1}f_{n-1}+ b_nf_n + a_nf_{n+1}\\ 
(Pf)_n &= -a_{n-1}f_{n-1} + a_nf_{n+1}\,.
\end{align*}
and recall that in the standard basis $L=L\big(\{a_n\}_{n\in\mathbb{Z}},\{b_n\}_{n\in\mathbb{Z}}\big)$ is a Jacobi matrix 
(symmetric, tridiagonal with positive off-diagonals) and $P$ is a bounded skew-adjoint operator, i.e., $P^* = -P$:
\begin{equation*}
L=
\begin{pmatrix} 
  \ddots    &\ddots & 	\ddots	 & 		   \\
   \ddots & b_{n-1} & a_{n-1} & 0           \\
 \ddots          & a_{n-1} & b_n    & a_n     & \ddots        	 \\
	   & 0         & a_n    & b_{n+1} & \ddots       	 \\
		   &		  & \ddots      &\ddots & \ddots         \\
\end{pmatrix}
\,,~
P=
\begin{pmatrix} 
  \ddots    &\ddots & 	\ddots	 & 		   \\
   \ddots & 0 & a_{n-1} & 0           \\
 \ddots          & -a_{n-1} & 0    & a_n     & \ddots        	 \\
	   & 0         & -a_n    &0 & \ddots       	 \\
		   &		  & \ddots      &\ddots & \ddots         \\
\end{pmatrix}.
\end{equation*}
The evolution equations for the old $p$ and $q$, and the new $a$ and $b$ variables can be found
through straightforward and short calculations.
The equations of motion induced by $\mathcal H_\varepsilon$ are
\begin{equation}\label{E:eom_pert_pq}
\begin{cases}
\partial_t{p}_n= e^{-(q_n-q_{n-1})}-e^{-(q_{n+1}-q_n)}+ \varepsilon \big[u'(q_{n+1} - q_n) - 
u'(q_n - q_{n-1})\big] \\
\partial_t{q}_n=p_n\,.
\end{cases}
\end{equation}
In terms of the $a,b$ variables, this translates to
\begin{equation}\label{E:eom_p_ab}
\begin{cases}
\partial_t{a}_n = a_n\big(b_{n+1} - b_n\big)\\
\partial_t{b}_n = 2\big(a_n^2 - a_{n-1}^2\big) + \frac{\varepsilon}{2} \big(c_{n-1} - c_n \big)\,
\end{cases}
\end{equation}
where
\begin{equation*}
c_n = u'\left(\log{\left(\frac{1}{4a_n^2}\right)}\right)~\text{ for all}~n\in\mathbb Z\,.
\end{equation*}
The evolution of the Jacobi matrix $L$ is given by a perturbation to the usual Lax pair of the Toda lattice:
\begin{equation}\label{E:eom_p_L}
\partial_t{L} = [P, L] + \varepsilon U\,,
\end{equation}
where $U$ is a diagonal matrix with
$U_{nn} =\frac{1}{2} \big( c_{n-1} - c_n \big)$ for all $n\in\mathbb Z\,.$
We note that one of the conservation laws associated to the Toda lattice still holds true; namely a straightforward calculation 
shows that $\tr(L-L_0)$ (where $L_0=L\big(\{\frac12\},\{0\}\big)$ is the free Jacobi matrix) is conserved in time by the perturbed 
evolution. It is however no longer true that traces of higher powers of $L-L_0$ are conserved, which is not surprising as we 
expect the perturbed dynamics to be non-integrable.

As is well-known, the integrability of the Toda lattice can be exploited via the
bijective correspondence between the Lax operator, which in this case is the Jacobi matrix $L$,
and its scattering data. This correspondence goes under the name of direct and inverse scattering
theory, and has already been studied in detail. While we do not attempt to give a comprehensive
survey of the relevant references, the interested reader may enter the subject, for 
example, through \cite{T}, \cite{Tesch01}, and the references therein.

Recall that, since $L$ is a bounded self-adjoint operator, we know that the spectrum $\sigma(L) \subset \mathbb{R}$ and $L$ 
has no residual spectrum. For details, see \cite{RS}. Furthermore, if the sequences $(a_n)_{n\in\mathbb{Z}}$ and $(b_n)_{n\in\mathbb{Z}}$ 
satisfy the hypotheses of Theorem~\ref{T:SpatialDecay} (\ref{A:app}) at time t=0 (and hence at all times $t\geq 0$), then one can further 
conclude that the spectrum of $L$ consists of a purely absolutely continuous part
$\sigma_{\rm{ac}}{(L)} = [-1,1]$
and a (finite) pure point part
\begin{equation*}
\sigma_{\text{pp}}{(L)} = \{ \lambda_j\, \colon j=1,\dots,N \}\subset (-\infty,-1)\cup(1,\infty)\,.
\end{equation*}
All the eigenvalues $\lambda_j$ are simple.
For convenience, we map the spectral data
via the so-called \textit{Joukowski transformation}:
$$
\lambda=\frac12\,\left(z+\frac1z\right),~\quad z=\lambda-\sqrt{\lambda^2-1},~\quad \lambda\in\mathbb{C},~|z|\leq 1\,. 
$$
Using this parameter $z$, it is standard to show that for any $0<|z|\leq 1$ there exist unique \textit{Jost 
solutions} $\varphi_{\pm}(z,n)$, i.e.\ solutions of 
\begin{equation}\label{E:defnJost}
L\varphi_{\pm}(z,\cdot)=\frac{z+z^{-1}}{2}\,\varphi_{\pm}(z,\cdot)\,,
\end{equation} 
normalized such that
\begin{equation}\label{E:normJost}
\lim_{n\to\pm\infty} \varphi_{\pm}(z,n)\cdot z^{\mp n}=1\,.
\end{equation}
The functions $z\mapsto z^{\mp n}\varphi_\pm (z,n)$ are holomorphic in the domain
$|z|<1$ and continuous on $|z|\leq 1$. If we focus on the unit circle $|z|=1$ with $z^2\neq 1$,
we observe that $\varphi_{\pm}(z,\cdot), \varphi_{\pm}(z^{-1},\cdot)$ are linearly independent,
and hence we can obtain the \textit{scattering relations}:
\begin{align}
T(z)\varphi_+(z,n) = R_{-}(z) \varphi_-(z,n) + \varphi_-(z^{-1}, n)\label{E:scatteringsoln1}\\
T(z)\varphi_-(z,n) = R_{+}(z) \varphi_+(z,n) + \varphi_+(z^{-1}, n)\label{E:scatteringsoln2}
\end{align}
for all $n\in\mathbb{Z}$, and for $|z|=1$, with $z^2 \neq 1$.

But more is true. Indeed, the transmission coefficient $T(z)$ has a meromorphic extension inside the 
entire unit disk $|z|\leq1$, with finitely many simple poles $\zeta_k \in (-1,0) \cup (0,1)$, $k=1, \dots, N$. 
The locations of the poles are related to the eigenvalues of the original Jacobi matrix through the 
Joukowski relation:
\begin{equation*}
\lambda_k=\frac{\zeta_k+\zeta_k^{-1}}{2},~\text{ for all }~k\in\{1,\dots,N\}\,.
\end{equation*}
Note that positive $\zeta_k$'s correspond to the eigenvalues above 1, while negative $\zeta_k$'s 
correspond to the eigenvalues below -1.
An important description of the locations of the poles of $T$ is as the points $z$ inside
the unit disk, $|z|<1$, where the Jost solutions $\varphi_+(z,\cdot)$ and $\varphi_-(z,\cdot)$
are constant multiples of each other -- and thus both in $\ell^2(\mathbb{Z})$. One can then
compute the residue of $T$ at each simple pole $\zeta_k$, as follows:
\begin{equation*}
\text{Res}(T;\zeta_k) = -\mu_k\zeta_k \gamma_{k,+} = -\, \frac{\zeta_k\gamma_{k,-}}{\mu_k}\,,
\quad\text{where }\gamma_{k,\pm} = \displaystyle\frac{1}{\big\|\varphi_{\pm}(\zeta_k, \cdot)\big\|_{\ell^2}^2}
\end{equation*}
are the \textit{norming constants}, and $\mu_k$ is the associated proportionality constant,
$\varphi_+(\zeta_k, \cdot) = \mu_k\varphi_-(\zeta_k, \cdot)\,.$

It is a fundamental fact of scattering theory
that Jacobi matrices $L$ whose coefficients decay fast enough (as the hypothesis of Theorem~\ref{T:SpatialDecay} in~\ref{A:app}) 
are in bijective correspondence 
with their scattering data $\{R,\zeta_k,\gamma_k\,|\,1\leq k\leq N\}$, where we use the standard convention of setting
\begin{equation*}
R(z) = R_{+}(z)~\text{ and }~\gamma_k=\gamma_{k,+}\,.
\end{equation*}
Note that implicit in this statement is the fact that from $R_+$ and the $\zeta_k$'s and $\gamma_{k,+}$'s one can fully reconstruct $R_-$, $T$ and the $\gamma_{k,-}$'s.
From this point onwards we will always use the notation above for $R$ and $\gamma_k$, unless specified otherwise.

Furthermore, if the Jacobi matrix $L$ evolves according to the Toda lattice, \eqref{E:eom_p_L} with $\varepsilon=0$, 
then the $\zeta_k$'s are constant, while $R$ and the $\gamma_{k}$'s satisfy simple linear evolution equations.
If $L$ is allowed to evolve according to the perturbed equation \eqref{E:eom_p_L}, then these linear equations
pick up perturbation terms (see \eqref{E:eom_p_R_+}, \eqref{E:eom_p_eigenvalues}, and \eqref{E:eom_p_norm})
which require full knowledge of not only $L(t)$, but also the scattering data and the Jost solutions. 
Due to the intricate structure of these evolution equations, it turns out to be easier, numerically, 
to compute the scattering data at a time $t$
directly from the  (truncation of the) Jacobi matrix $L(t)$. In the
simplest case, that of the evolution of the eigenvalues, we numerically compare the results of our direct calculations
with the evolution equation \eqref{E:eom_p_eigenvalues}. The results of this comparison can be found on
the website of the project\footnote{Project website: \texttt{http://bilman.github.io/toda-perturbations}}.

\section{The Numerical Scheme}\label{S:scheme}
Before moving on to presenting the numerical results of our work, we describe the numerical scheme used for solving \eqref{E:eom_p_L} and computing scattering data associated to the (doubly-infinite) Jacobi matrix $L$. 
\subsection{Time-stepping for computing $L(t)$}
We note that $L$, $P$, and $U$ are discrete operators. Therefore, there is no need for spatial discretization to solve \eqref{E:eom_p_L} numerically. To approximate the solutions we truncate the doubly-infinite lattice at particles with indices $\pm N\in\mathbb{N}$ for some large $N$, and work with the truncated matrices $L_N$
\begin{equation*}
L_N=
\begin{pmatrix} 
b_{-N} & a_{-N}    &   	0	   		 &             &  \\ 
a_{-N} & b_{-N+1} & a_{-N +1} 		 & \ddots		 &  \\
0         &  a_{-N +1} &\ddots  & \ddots         & \ddots& \\
          &	\ddots	   & \ddots              & \ddots & a_{N-2}    & 0     	 \\
           &		   &		  \ddots         & a_{N-2} & b_{N-1}   & a_{N-1}       \\
          &		   &		   		 & 	0	 &a_{N-1}		 & b_N 	     \\
\end{pmatrix}\,,
\end{equation*}
and accordingly, $P_N$. Since the solutions $(a,b)$ of \eqref{E:eom_p_ab} satisfy
\begin{equation*}
\big( a(t) - \tfrac{1}{2}, b(t) \big) \in \ell^2(\mathbb{Z}) \oplus \ell^2(\mathbb{Z})~\text{ for all}~t \geq 0\,,
\end{equation*}
we close the finite system of differential equations for the truncated system by imposing the Dirichlet boundary conditions given by
\begin{equation}\label{E:bc}
a_{-(N+1)} - \frac{1}{2} = a_{N} - \frac{1}{2} = 0~\text{ and }~b_{\pm(N+1)} = 0\,,
\end{equation}
and consider
\begin{equation}\label{E:eom_finite}
\begin{cases}
\partial_t{a}_n = a_n\big(b_{n+1} - b_n\big)\\
\partial_t{b}_n = 2\big(a_n^2 - a_{n-1}^2\big) + \frac{\varepsilon}{2} \big(c_{n-1} - c_n \big)\,,
\end{cases}
\end{equation}
for $n\in\{-N,-N+1,\dots,N-1,N\}$, subject to the boundary conditions given in \eqref{E:bc}. 

To integrate \eqref{E:eom_finite}, we adopt the $4^{\text{th}}$-order Runge-Kutta time-stepping method. We define the temporal discretization error $E_{\Delta t}$ at time $t$ by
\begin{equation*}
E_{\Delta t}(t) = \left \| L_N (t) - L_N^* (t) \right\|_{HS}\,,
\end{equation*}
where $\Delta t$ is the step size for temporal discretization and $\| \cdot \|_{HS}$ is the Hilbert-Schmidt norm. When an exact solution is available at hand, $L_N^*$ stands for the finite truncation of the exact solution to the infinite dimensional problem. Otherwise, $ L_N^*$ is taken to be the solution obtained by choosing $\Delta t$ extremely small. Table \ref{T:RungeKuttaConv} lists the discretization errors measured at $t=320$ and demonstrate that
\begin{equation*}
\alpha=\log_2 \left | \frac{E_{2\Delta t}(t)}{E_{\Delta t}(t)}\right |
\end{equation*}
approaches 4, the order of the time-stepping method, as the temporal step-size is diminished by half. The underlying experiment is comprised of pure 1-soliton initial data ($k=0.4$) which is let to evolve in the Toda lattice ($\varepsilon = 0$) and a perturbed lattice with with $u(r)=r^2$, $\varepsilon=0.05$.
\begin{table}[htp]
\begin{center}
\small
\begin{tabular}{|c|c|c|}
\cline{1-3}
$\Delta t$ & $E_{\Delta t}(T)$, $\varepsilon=0$ & $\alpha$  \\
[0.5ex]
\hline
$4h$ & 1.3438e-06 & --\\
[1ex] 
$2h$ & 8.1798e-08 & 4.0381 \\
[1ex] 
 $h$ & 5.0433e-09 & 4.0196 \\
 [1ex] 
 $\frac{h}{2}$ & 3.1303e-10 & 4.0099 \\
 [1ex] 
 \hline
\end{tabular} 
\qquad
\begin{tabular}{|c|c|c|}
\cline{1-3}
$\Delta t$ & $E_{\Delta t}(T)$, $\varepsilon=0.05$ & $\alpha$  \\
[0.5ex]
\hline
 $4h$ & 1.4236e-06 & --\\
[1ex] 
 $2h$ & 8.6975e-08 &4.0328 \\
[1ex] 
 $h$ &5.3727e-09 &4.0168 \\
 [1ex] 
 $\tfrac{h}{2}$ & 3.3363e-10 &4.0093\\
 [1ex] 
 \hline
\end{tabular} 
\caption{Temporal discretization errors and $\alpha$ measured at time $T=320$ for one-soliton initial data, $k=0.4$; $h=0.02$, and $N=2^{13}$ }
\label{T:RungeKuttaConv}
\end{center}
\end{table}%
Figure~\ref{F:Error} displays the error growth in time.
\begin{figure}[h!]\centering
\subfigure[$E_{\Delta t }(t)$]{
\includegraphics[width=188pt]{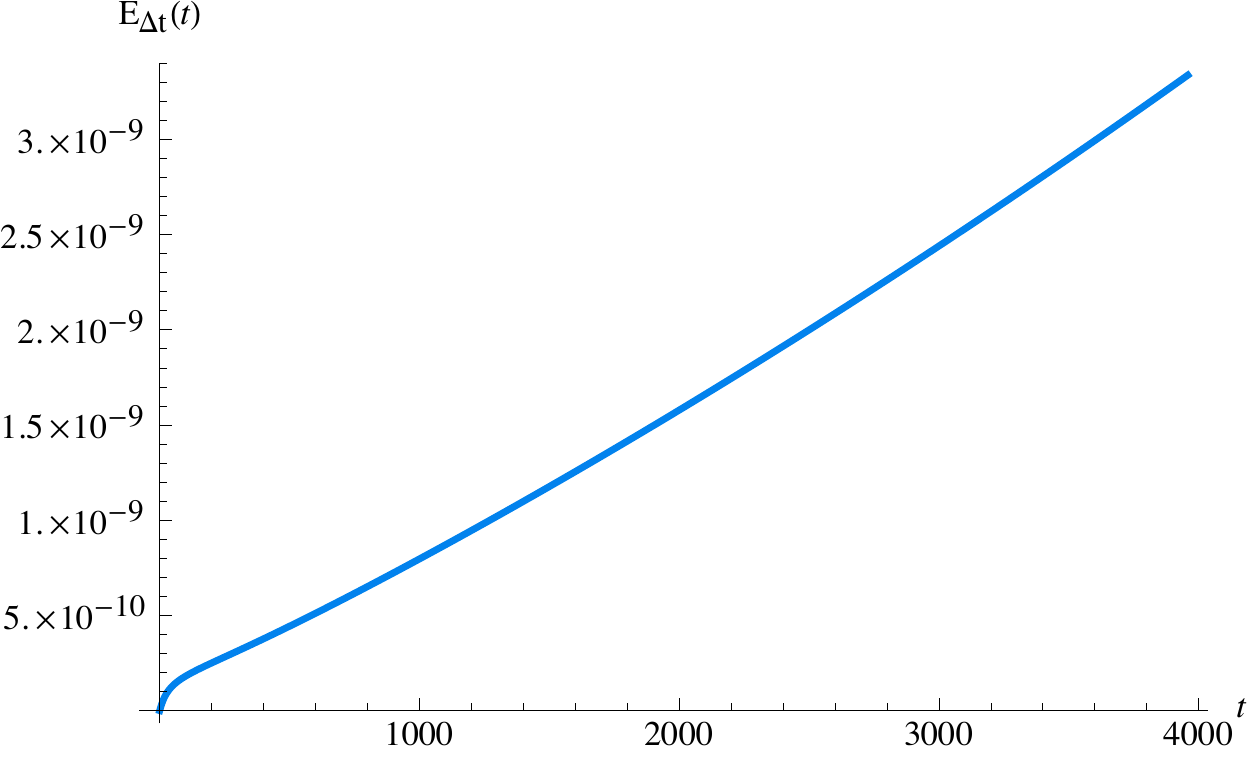}
}
~
\subfigure[$\log\left(E_{\Delta t}(t)\right)$]{
\includegraphics[width=188pt]{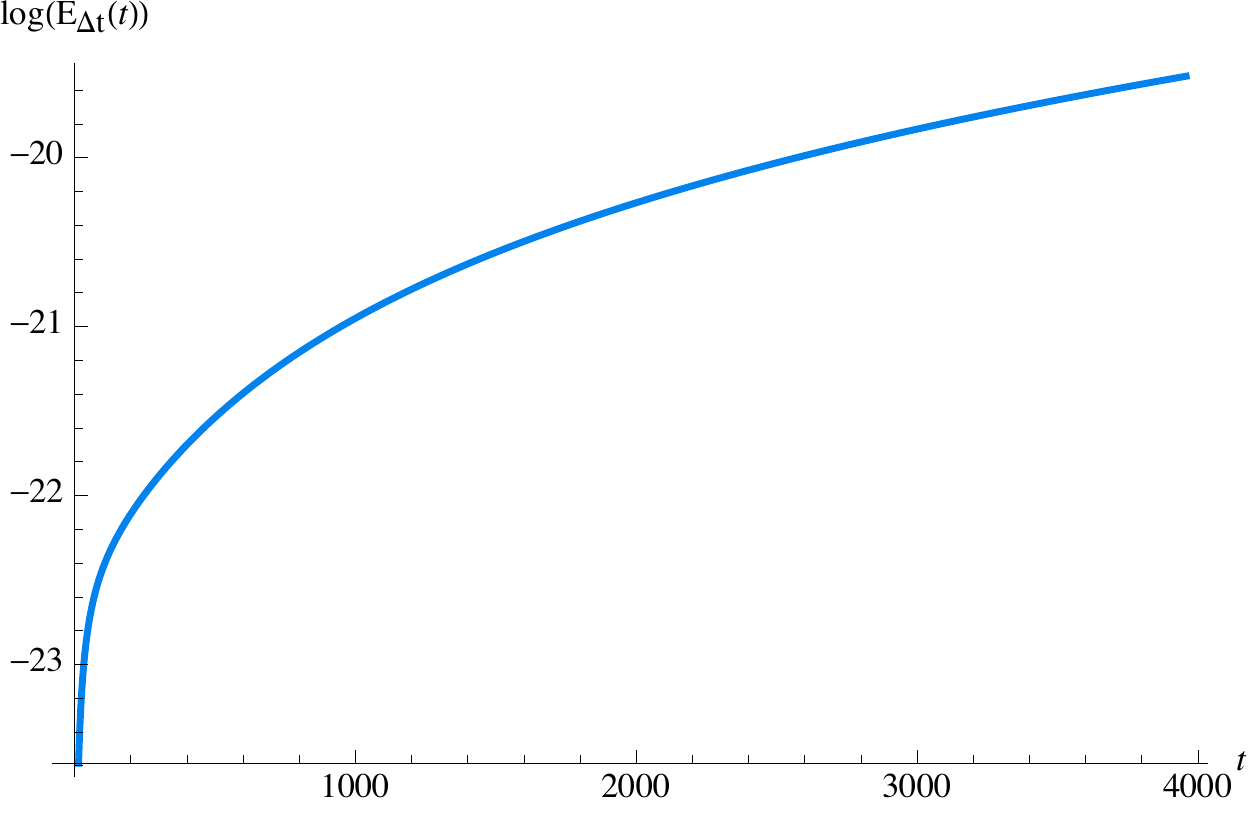}
}
\caption{$E_{\Delta t}(t)$ and $\log\left(E_{\Delta_t}(t)\right)$, with $\Delta t = 10^{-2}$ and $N=2^{14}$.}
\label{F:Error}	
\end{figure}
\subsection{Computing eigenvalues of $L(t)$}
We approximate the discrete spectrum of $L(t)$ by computing the eigenvalues of its truncation $L_N(t)$ using the QR-algorithm (provided by \texttt{LAPACK}, \cite{lapack}) at each time-step. Since $L$ is a discrete operator, no spatial discretization is used to obtain the matrix $L_N$. To justify accuracy of this method, we begin with some general facts concerning eigenvalues of doubly-infinite (whole-line) Jacobi matrices and eigenvalues of their finite truncations. We omit the proofs of these facts and refer the reader to the articles they are taken from. As in \cite{HS}, given any bounded self-adjoint operator $A$ on $\ell^2(\mathbb{Z})$, define
\begin{equation}\label{E:variational}
\begin{aligned}
\lambda^+_j (a) &= 	\inf_{\phi_1,\dots, \phi_{j-1}}\sup_{\substack{\psi\colon \psi \perp \phi_j \\ \psi\in\ell^2(\mathbb{Z}), \| \psi \| = 1}} \langle \psi, A \psi \rangle\\
\lambda^-_j (a) &= 	\sup_{\phi_1,\dots, \phi_{j}}\inf_{\substack{\psi\colon \psi \perp \phi_j \\ \psi\in\ell^2(\mathbb{Z}), \| \psi \| = 1}} \langle \psi, A \psi \rangle\,.
\end{aligned}
\end{equation}
From the definitions, we have
\begin{equation*}
\lambda^-_1 (a) \leq \lambda^-_2 (a) \leq \cdots \leq \lambda^+_2(a) \leq \lambda^+_1 (a)\,.
\end{equation*}
\begin{prop}[p. 111, \cite{HS}]
The min-max principle \cite[Theorem XIII.1]{RS4} asserts that
\begin{itemize}
\item[(1)] $\lim_{n\rightarrow +\infty} \lambda^{+}_n (a) = \sup \sigma_{\textrm{ess}}(a)$ and $\lim_{j\rightarrow +\infty} \lambda^{-}_j (a) = \inf \sigma_{\textrm{ess}}(a)$.
\item[(2)] If $A$ has $M^{+}$ eigenvalues counting multiplicity above the essential spectrum, these eigenvalues are precisely $\lambda^{+}_1, \lambda^{+}_2, \dots \lambda^{+}_{M^+}$, and $\lambda^{+}_j = \sup \sigma_{\textrm{ess}}(a)$ for $j>M^+$.
\item[(3)] If $A$ has $M^{-}$ eigenvalues counting multiplicity below the essential spectrum, these eigenvalues are precisely $\lambda^{-}_1, \lambda^{-}_2, \dots \lambda^{-}_{M^-}$, and $\lambda^{-}_j = \sup \sigma_{\textrm{ess}}(a)$ for $j>M^-$.
\end{itemize}
\end{prop}

The following result describes the impact of truncations on eigenvalues of Jacobi matrices.
\begin{prop}[Proposition 2.2, \cite{HS}]\label{P:EV1}
Let $\Pi$ be an orthogonal projection, and $A$ be a bounded self-adjoint operator on $\ell^2(\mathbb{Z})$. Define $A_{\Pi} = \Pi A \Pi$, restricted as an operator onto the range of $\Pi$. Then
\begin{equation*}
\lambda^{+}_j (A_{\Pi}) \leq \lambda^{+}_j (a)~\text{ and }~\lambda^{-}_j (A_{\Pi}) \geq \lambda^{-}_j (a)\,,
\end{equation*}
for $n > 0$.
\end{prop}
\begin{proof}
Changing from $A$ to $A_\Pi$ in \eqref{E:variational} adds the condition $\psi \in \text{Ran}\, \Pi$. This increases infimums and decreases supremums, hence gives us the desired inequalities.
\end{proof}
\begin{remark}\label{R:new_eigs}
Let $\lambda^{-}_1 < \lambda^{-}_2 < \cdots <\lambda^{-}_{2N+1}$ and $\lambda^{+}_{2N+1} < \dots <\lambda^{+}_ {2} <\lambda^{+}_{1}$ denote the real simple eigenvalues of $L_N$, labeled in increasing and decreasing order, respectively. Suppose that $L$ has $M^{-}$ eigenvalues below its ac spectrum and $M^{+}$ eigenvalues above its ac spectrum. Note that $M^{\pm}$ are finite. Since $L$ is bounded, the quadratic form of $L_N$ is a restriction of the quadratic form of $L$ to $\mathbb{C}^{2N+1}$. Then by Proposition \ref{P:EV1}, for any $N\in\mathbb{N}^{+}$, we have
\begin{equation}\label{E:EVtrunc1}
\begin{aligned}
\lambda^{-}_j (L) &\leq \lambda^{-}_j (L_{N})&&\quad\text{for}\phantom{x} j =1,2,\dots,\min(2N+1, M^{-})\\
\lambda^+_{j} (L_N) &\leq \lambda^{+}_j (L)&&\quad\text{for}\phantom{x} j =1,2,\dots,\min(2N+1, M^{+})\,,
\end{aligned}
\end{equation} 
as in \cite[Theorem XIII.3]{RS4}. Furthermore, 
\begin{equation*}
\begin{aligned}
-1 &\leq \lambda^{-}_j (L_{N})&&\quad\text{for}\phantom{x}  \min(2N+1, M^{-})<j \leq 2N+1\\
\lambda^{+}_{j} (L_N) &\leq 1 &&\quad\text{for}\phantom{x} \min(2N+1, M^{+})< j \leq 2N+1\,.
\end{aligned}
\end{equation*}
This implies that if $L_N$ has $M$ eigenvalues that are strictly less than $\inf\sigma_{\text{ac}}(L)=-1$, then $L$ has at least $M$ eigenvalues below its continuous spectrum. Analogously, if $L_N$ has $M$ eigenvalues that are strictly greater than $\sup\sigma_{\text{ac}}(L) = 1$, then $L$ has at least $M$ eigenvalues above its continuous spectrum. For a more detailed discussion on this matter, we refer the reader to \cite[Theorem XIII.3]{RS4}, the discussion after Proposition 2.4 in \cite{Sim2}, or \cite{HS}.
\end{remark}

More is true regarding eigenvalues of finite truncations of $L$. Since eigenvalues of $L_N$ and eigenvalues of its principal sub matrix strictly interlace, \cite[Proposition 2.1]{Sim2},
\begin{equation}\label{E:EVtrunc3}
\lambda^{-}_j(L_{N+1}) < \lambda^{-}_j(L_N)~\text{and}~ \lambda^{+}_j(L_{N}) < \lambda^{+}_j(L_{N+1}),~\text{for}~ j=1,\dots,2N+1\,.
\end{equation}
\eqref{E:EVtrunc1} and \eqref{E:EVtrunc3} together imply that eigenvalues of $L_N$ have limit points outside $[-1,1]$ as $N\rightarrow +\infty$. Note that, even for self-adjoint operators, it is not in general true that the set of these limit points is equal to the pure point spectrum of the operator under study. However, in our case, $L$ has additional properties. First, finite truncations of $L$ are also self-adjoint operators. Second, at any time $t\geq 0$, $L(t)$ is a compact perturbation of the discrete free Schrr\"{o}dinger operator $L^0$ that has $a_n \equiv \tfrac{1}{2}$, $b_n \equiv 0$, and $\sigma\big(L^0\big) =\sigma_{\text{ac}}\big(L^0\big) = [-1,1]$. As proven in \cite{IP}, these two facts imply:
\begin{prop}[Theorem 2.3, \cite{IP}]\label{P:EVconv}
Let $Z(L)$ denote the set of all limit points of $\bigcup_{N=1}^{\infty} \sigma\big({L_N}\big)$. Then $Z(L) = \sigma(L)$.
\end{prop}
Proposition~\ref{P:EVconv} asserts that eigenvalues of $L_N$ that are outside $[-1,1]$ converge to eigenvalues of $L$ as $N\rightarrow +\infty$. Let $E_{\lambda_j}(t)$ denote the error in computing an eigenvalue $\lambda_j$ of $L$, defined by
\begin{equation*}
E_{\lambda_j}(t) = |\lambda_j(L_N(t)) - \lambda_j^*(t)|,
\end{equation*}
where $\lambda_j$ is the computed eigenvalue of a finite truncation $L_N$, and $\lambda_j^*$ is the eigenvalue $\lambda_j(L(t))$ of the doubly-infinite operator $L$.
\begin{table}[htp]
\centering
\small
\begin{tabular}{|c|c|c|c|c|}
\cline{1-5}
$\Delta t$ & $E_{\lambda_1}(T)$, $\varepsilon=0$ & $E_{\lambda_1^{-}}(T)$, $\varepsilon=0.05$ & $E_{\lambda_2^{-}}(T)$, $\varepsilon=0.05$ & $E_{\lambda_1^{+}}(T)$, $\varepsilon=0.05$ \\[1ex]
\hline
$4h$ & 1.262224e-09 & 5.579751e-08&1.1505353e-08& 3.370637e-12\\
$2h$ & 3.945577e-11 &  2.306966e-09& 6.701292e-10 & 1.567634e-13
\\
$h$ & 1.227241e-12 &  1.071940e-10& 4.037126e-11 &1.865174e-14\\
 $\frac{h}{2}$ & 3.042011e-14 & 5.523581e-12& 2.511102e-12& 2.620126e-14\\
 $\frac{h}{4}$ & 1.088019e-14 & 2.873257e-13& 1.869615e-13& 2.353672e-14\\
 \hline
\end{tabular} 
\caption{ $h=0.01$, $T=320$, for 1-soliton initial data with $k=0.4$, $N=2^{13}$.}
\label{T:Eigs}
\end{table}
Table~\ref{T:Eigs} displays of $E_{\lambda_j}(T)$ for one-soliton initial data, measured at $T=320$ in the Toda lattice and in the 
perturbed lattice with $u(r)=r^2$, $\varepsilon =0.05$. For one-soliton solution of the Toda lattice, the exact eigenvalue 
$\lambda_1^*$ is known, $\lambda_1^*(t)=\pm\cosh(k)$, for $t\geq 0$. In experiments with perturbed lattices, $\lambda_j^*(t)$ is 
computed from the reference solution $L^*(t)$. Figure~\ref{F:EigError} displays $E_\lambda (t)$ over time in the Toda lattice for 
1-soliton data with $k=0.6$, where $\Delta t=0.01$. 
\begin{figure}[h]\centering
\includegraphics[width=188pt]{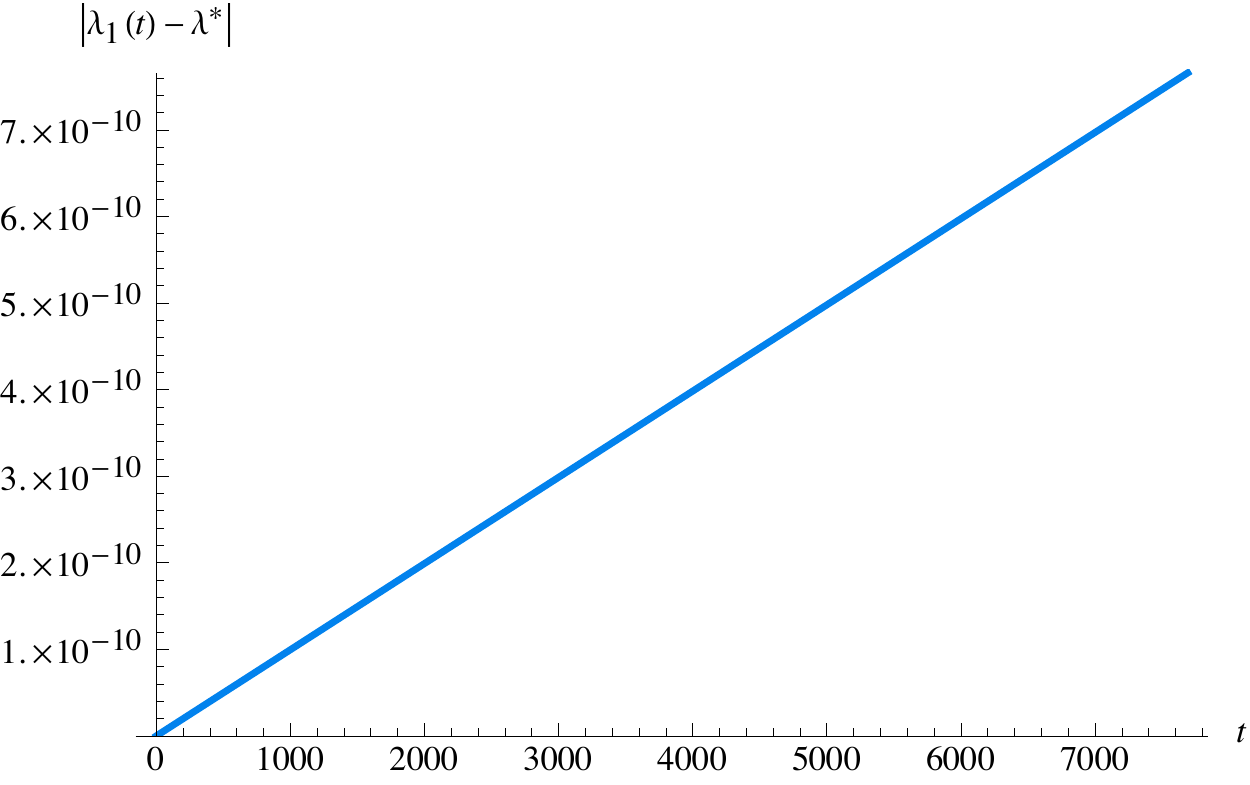}
\caption{$E_{\lambda_1}(t)$, with $\Delta t = 10^{-2}$ and $N=2^{13}$ for 1-soliton data with $k=0.6$; $\varepsilon=0$.}
\label{F:EigError}	
\end{figure}
\subsection{Computing the reflection coefficient $R(z;t)$}
The numerical procedure we adopt for approximating the reflection coefficient in scattering data associated to $L(t)$ at a time $t$ 
is analogous to the one carried out for the Korteweg-deVries equation in \cite{Tom_KdV} and for the nonlinear Schr\"{o}dinger equation in \cite{Tom_NLS}. We seek solutions of the form 
\eqref{E:normJost} to \eqref{E:defnJost}. We define two new functions
\begin{equation}\label{E:getjost}
f(z;n,t) = \varphi_{+}(z;n,t)z^{-n}\,,\quad g(z;n,t) = \varphi_{-}(z;n,t)z^{n}\,,
\end{equation}
so that we have $f(z;n,t)\rightarrow 1$ as $n\rightarrow +\infty$ and $g(z;n,t)\rightarrow 1$ as $n\rightarrow -\infty$. 
Then $f$ solves
\begin{equation}\label{E:f}
\frac{a_{n-1}}{z}f(z;n-1) + \left(b_n - \frac{z+z^{-1}}{2}\right)f(z;n) + a_n z f(z;n+1) = 0\,,
\end{equation}
and $g$ solves
\begin{equation}\label{E:g}
a_{n-1}z g(z;n-1) + \left(b_n - \frac{z+z^{-1}}{2}\right)g(z;n) + \frac{a_{n}}{z}g(z;n+1) = 0\,.
\end{equation}
For any $z$, \eqref{E:f} can be solved for $f$ on $n\geq 0$ and \eqref{E:g} can be solved for $g$ on $n\leq 0$ by backward 
substitution method using the appropriate boundary conditions at infinity. Then \eqref{E:getjost} is inverted to recover 
$\varphi_{\pm}$, and the solutions are matched at $n=0$ to extract the reflection coefficient $R(z;t)$. The effect of time-stepping 
error in $L$ on this computation is measured in the Toda lattice by
\begin{equation*}
E_\text{R}(t) = \sqrt{\sum_{z\in\mathbb{T}}\big| R(z;t) - R(z;0)e^{(z-z^{-1})t}\big|^2}\,,
\end{equation*}
where $\mathbb{T}$ is a mesh on the unit circle, typically with $1000$ mesh points. In Section~\ref{S:ref}, reflection coefficient 
for clean solitary waves is computed with the choice of $\Delta t = 0.001$.
\begin{table}[htp]
\begin{center}
\small
\begin{tabular}{|c|c|}
\cline{1-2}
\rule{0pt}{2.5ex}    
$\Delta t$ & $E_{R}(T)$, $\varepsilon=0$, $N=2^{13}$  \\
[0.5ex]
\hline
 $8h$ & 1.610568e-06 \\
 $4h$ & 1.006489e-07 \\
 $2h$ &6.429764e-09  \\
 $h$ & 4.305598e-09 \\
 \hline
\end{tabular}
\end{center}
\caption{$E_{\text{R}}(T)$ at $T=1280$, in the Toda lattice; $h=0.001$.}
\label{T:ref}
\end{table}

The underlying experiment for the measurements displayed in Table~\ref{T:ref} takes off with the initial data
\begin{equation*}
a_n=\frac{1}{2}+\frac{1}{10}e^{-n^2},\quad b_n = \frac{1}{10\cosh(n)}\,,
\end{equation*}
for which the corresponding reflection coefficient is not identically zero. Finally, we note that reflection coefficients for clean 
solitary waves in Section~\ref{S:clean} are computed with the choice $\Delta t =0.001$.

As explained in Section~\ref{S:Background}, we consider perturbations $u$ and parameters $\varepsilon>0$ such
that $V_\varepsilon$ satisfies Assumption~\ref{A}.
The results of numerical experiments exhibit no qualitative difference in behavior of the solutions subject to the 
different choices of perturbations described in Section~\ref{S:Background}, and so in what follows we only 
include the results of experiments with $u(r)=r^2$ and $u(r)=r^3$, as 
we find those to be more convenient to present here. Results for a variety of perturbation 
functions are available on the website of this 
project\footnote{\label{foot_web}Project website: \texttt{http://bilman.github.io/toda-perturbations}}. 
Throughout the remainder of the paper, 
we refer to $(a,b)$ as solutions, since the direct and inverse scattering transforms for Jacobi matrices are carried out in these 
coordinates. For 1-soliton data at $t=0$ we use
\begin{equation*}
a_n = \frac{\sqrt{1+e^{-2k(n-1) }}\sqrt{1+e^{-2k(n+1)}}}{2\left( 1 + e^{-2kn}\right)}
~{\text{~and~}}~
b_n = \frac{e^{-k} - e^k}{2}\Bigg( \frac{e^{-2kn}}{1+e^{-2kn}} -  \frac{e^{-2k(n-1)}}{1+e^{-2k(n-1)}}\Bigg)\,,
\end{equation*}
which are obtained from $q_n = \log{\left(\frac{1+e^{-2k(n-1) }}{1+e^{-2kn}}\right)}$
as in \eqref{E:1sol_q} with $\gamma = 1 - e^{-2k}$. Behavior of solutions $a$ and $b$ under perturbed dynamics 
are found to be qualitatively identical. Therefore, we present the results only for the sequence $a.$ Results including both of the 
variables $a$ and $b$ are also available on the website of this project\footnotemark[2].

\section{Numerical results: Soliton Initial Data}
We proceed with results of numerical experiments, where we commence with initial data that is a pure Toda soliton and let it evolve under the perturbed dynamics.
\subsection{Emerging solitary waves}\label{S:soliton_initial}
We find that a leading solitary wave emerges from the soliton initial data, followed by a dispersive tail, and that a secondary, counter-propagating wave is generated as soon as $t$ becomes positive. The leading solitary wave is wider and it has smaller amplitude compared to the soliton initial data. In Figure~\ref{F:SolWave_x2_takeoff}, from $t=0$ to $t=25$, we see a typical occurrence of this phenomenon. The peak of the solution has been truncated to show the details of the dispersion. Figure~\ref{F:SolWave_x2_takeoff}(b) displays the secondary wave propagating towards left, and the bottom portion of the emerging solitary wave propagating towards right, with a dispersive tail that is under development.
\begin{figure}[h]\centering
\subfigure[$t=0$]{\includegraphics[width=188pt]{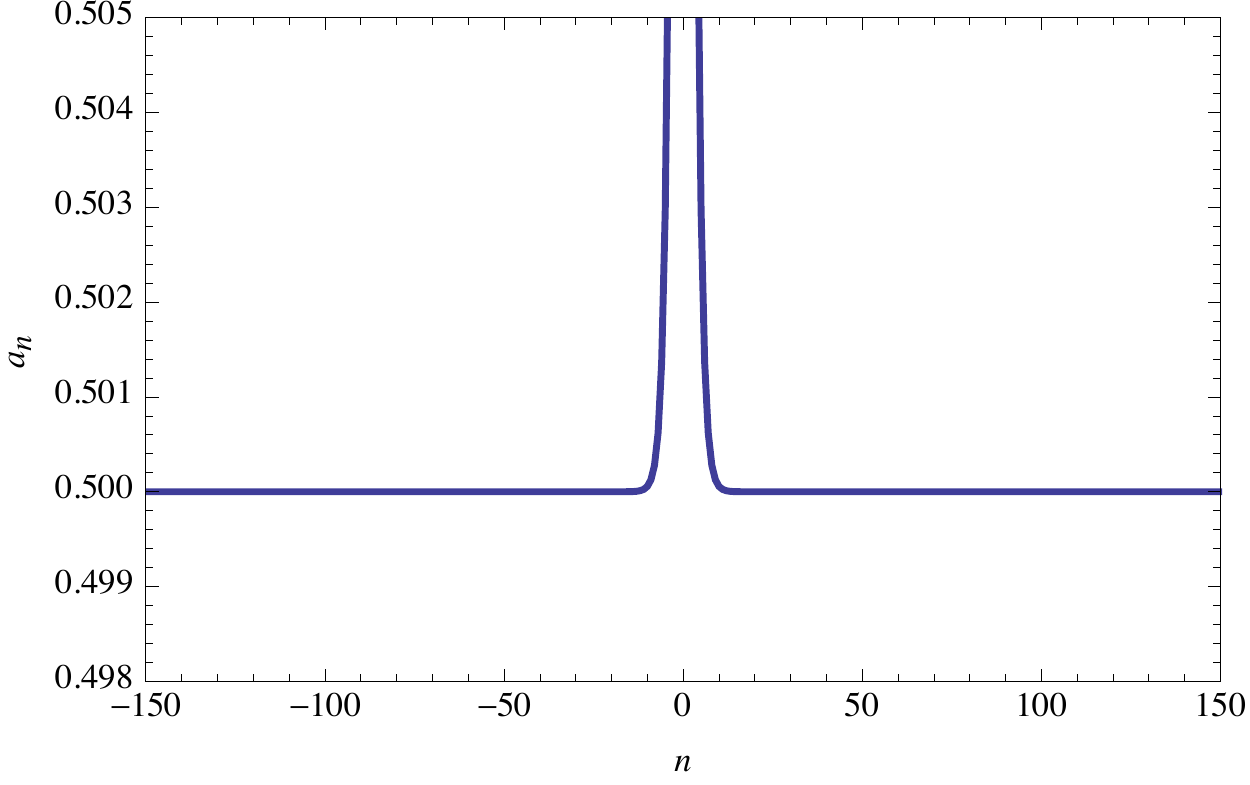}}\quad
\subfigure[$t=25$]{\includegraphics[width=188pt]{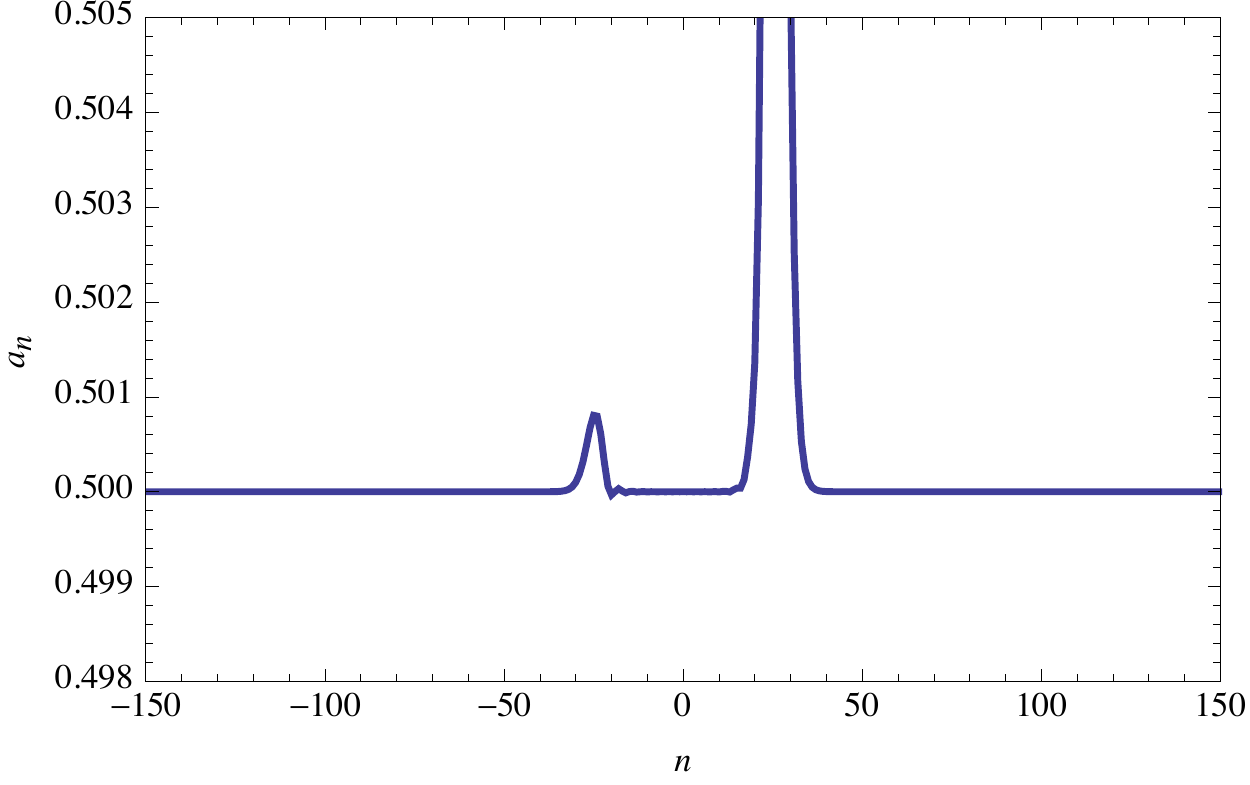}}
\caption{One-soliton initial data in the lattice with $u(r)=r^2$ and $\varepsilon=0.05$, $N=2^{13}$.}
\label{F:SolWave_x2_takeoff}	
\end{figure}
\begin{figure}[h]\centering
\subfigure[t=3000]{\includegraphics[width=188pt]{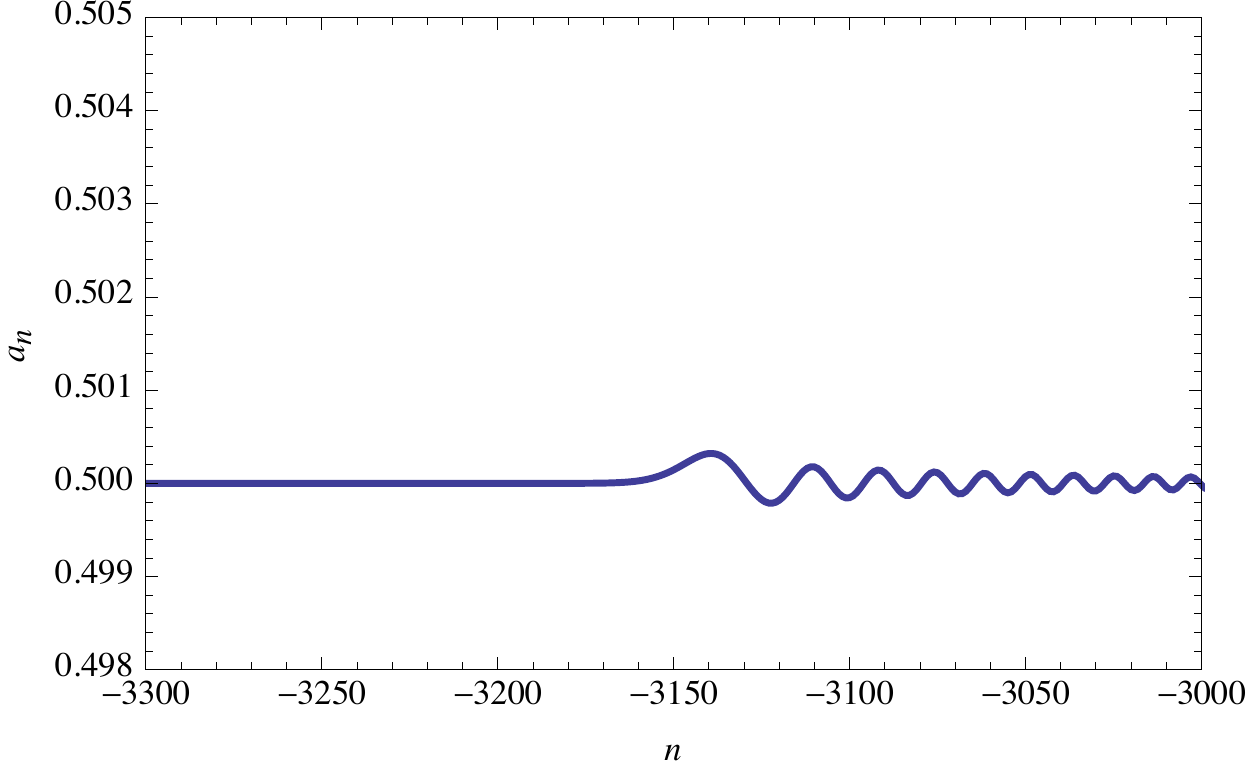}}\quad
\subfigure[t=3000]{\includegraphics[width=188pt]{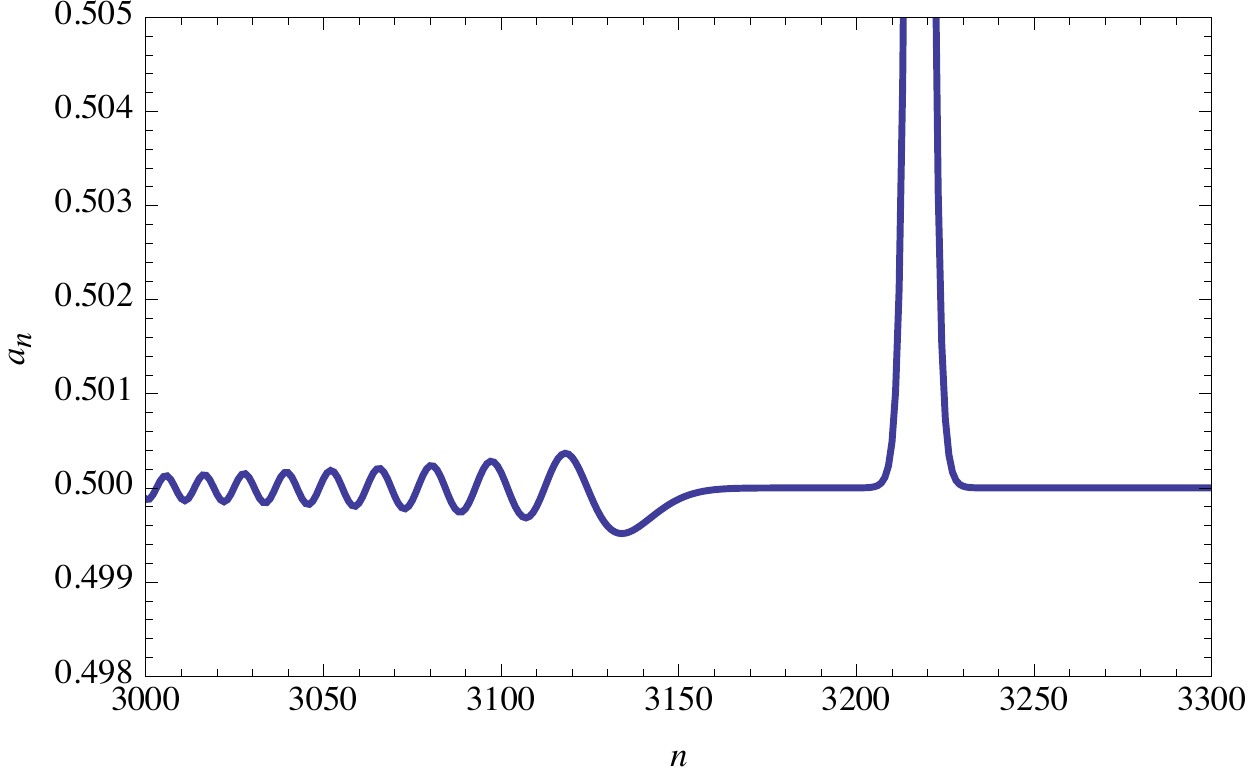}}
\caption{Solution in longer time scale, $u(r)=r^2$ and $\varepsilon=0.05$, $N=2^{13}$.}
\label{F:SolWave_x2_longtime}
\end{figure}
Figure~\ref{F:SolWave_x2_longtime} displays the solution in the same experiment at a later time $t=3000$. As can be seen in Figure~\ref{F:SolWave_x2_longtime}(a), the leading solitary wave gets separated from the dispersive tail as time elapses. Figure~\ref{F:SolWave_x2_longtime}(b) shows the counter-propagating wave which spreads and loses amplitude over the course of the entire numerical experiment. We find that this phenomena occurs in all of the perturbed systems we consider.

We now turn our attention to the numerical experiments in which we compare solitary waves that emerge from the same initial data under different perturbations. Figure~\ref{F:1sol-x2_x3} presents the leading solitary waves (in solid-red) plotted against the soliton solution of the Toda lattice (in dashed-blue), at time $t=3000$. 
\begin{figure}[h]
\centering
\subfigure[$u(r)=r^2$]{\includegraphics[width=188pt]{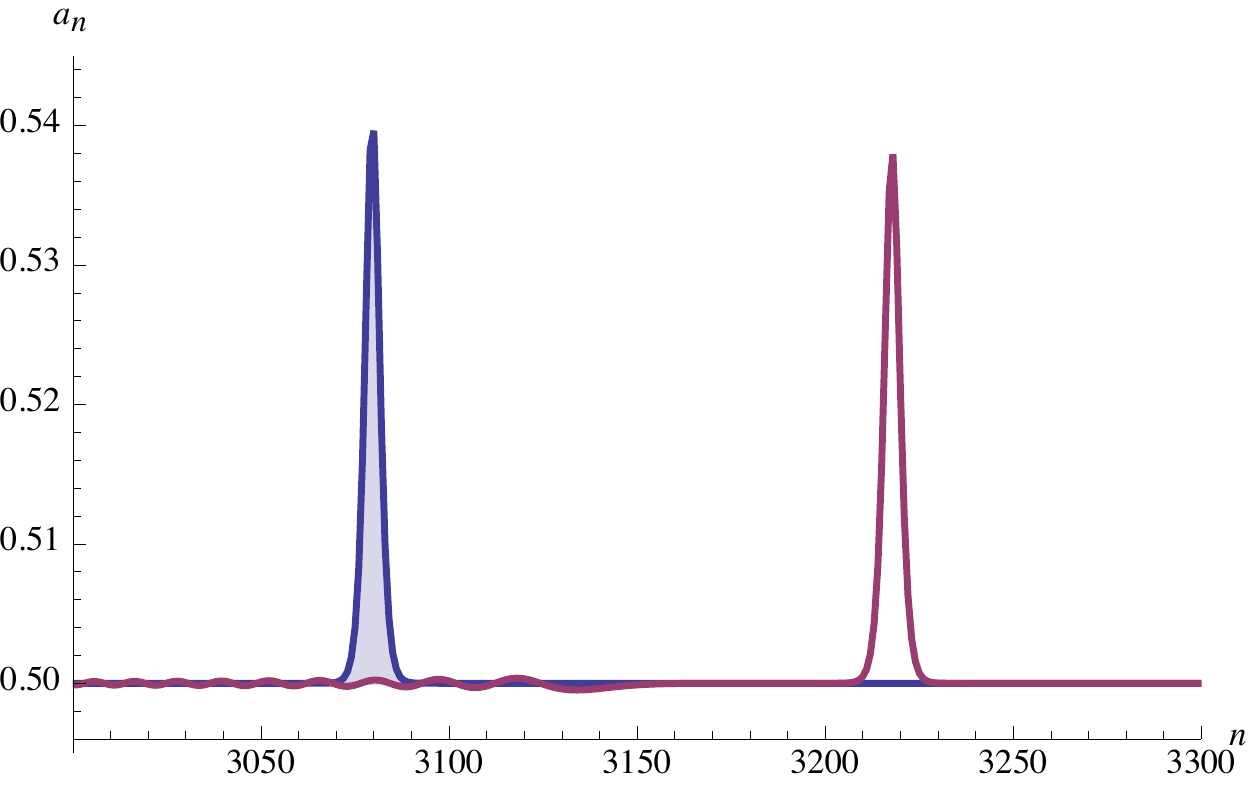}}\quad
\subfigure[$u(r)=r^3$]{\includegraphics[width=188pt]{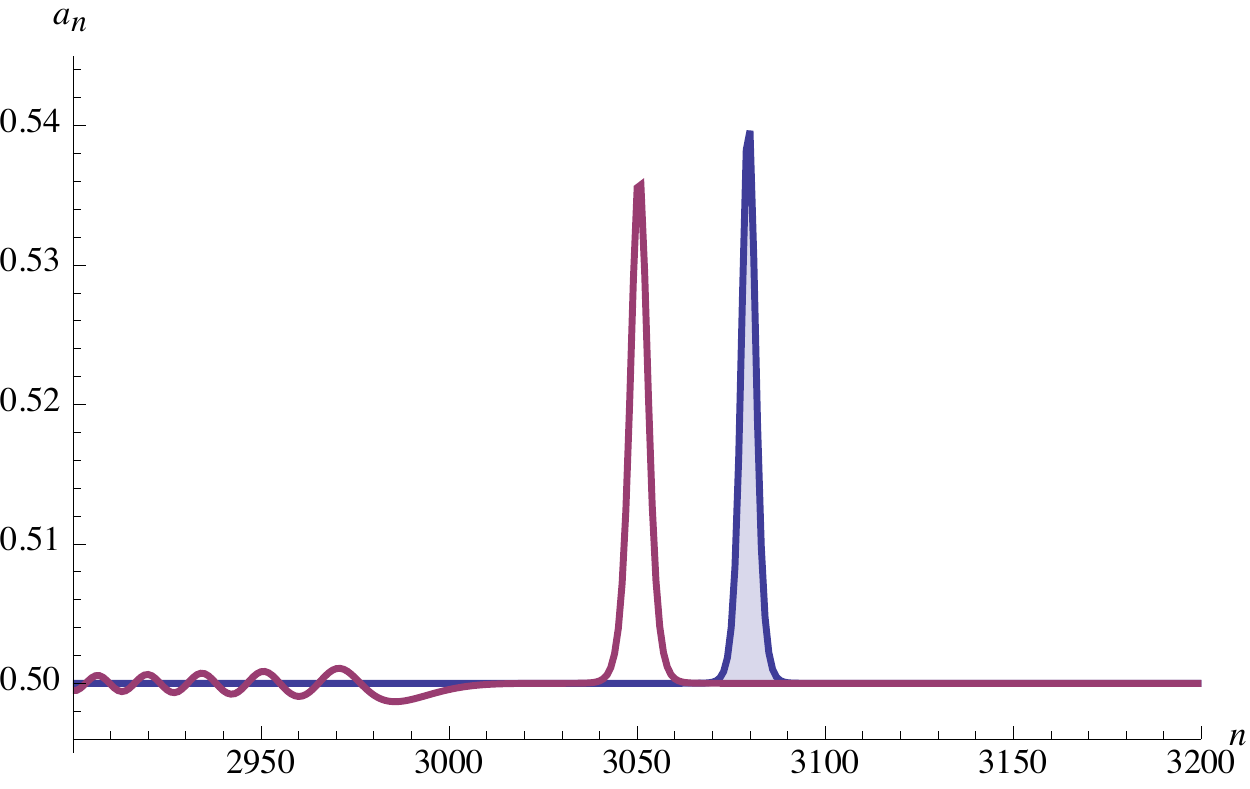}}
\caption{1-soliton initial data evolving under perturbed dynamics (solid purple) and under the Toda dynamics (filled blue) for different choices of perturbations, at $t=3000$.}
\label{F:1sol-x2_x3}
\end{figure}

As mentioned earlier, both in Figure~\ref{F:1sol-x2_x3}(a) and (b), amplitude of the leading solitary wave is smaller than that of the Toda soliton it emerges from. On the other hand, the potential with perturbation $u(r)=r^2$ generates a solitary wave that travels faster than the Toda soliton, whereas the potential with $u(r)=r^3$ yields a wave that travels slower than the soliton. This is also not surprising, since a perturbation in the interaction potential introduces a perturbation in the nonlinear dispersion relation \eqref{E:speed} that governs the propagation speed of these solitary waves, as we shall see further in this section. Note that, for a fixed perturbation, taller solitary waves travel faster than shorter ones do.

In Table ~\ref{T:1}, we list speed and amplitude pairs that are measured for the leading solitary waves in the same experiments. We measure the speed of propagation by tracking the peak of the solitary wave profile (in the continuous background) over time. To find the peak, we use fourth order polynomial interpolation. Small fluctuations in these computations are due to the fixed spatial grid size of the problem. This method is accurate up to $10^{-4}$ when an error check is performed against the exact height of solitons in the Toda lattice.

\begin{table}[htp]
\begin{center}
{\small
\begin{tabular}{c|c|c|c|c|c|c|}
\cline{2-7}
\small
& \multicolumn{2}{ c | }{\raisebox{-0.4ex}{$u(r)=r^2$}} & \multicolumn{2}{ c| }{\raisebox{-0.4ex}{$u(r)=r^3$}} &  \multicolumn{2}{ c| }{\raisebox{-0.4ex}{Toda}}\\[0.8ex] \cline{1-7}
\multicolumn{1}{|c|}{Time} & Amplitude & Speed & Amplitude & Speed & Amplitude & Speed \\ \cline{1-7}
\multicolumn{1}{ |c| }{5000} & 0.0381 &1.0737& 0.0365 & 1.0168 & \multirow{10}{*}{0.0405} & \multirow{10}{*}{1.0268} \\ \cline{1-5}
\multicolumn{1}{ |c  }{5100} &
\multicolumn{1}{ |c| }{0.0381} & 1.0731 & 0.0365& 1.0172 & \multicolumn{1}{ c| }{} &\multicolumn{1}{ c|  }{} \\ \cline{1-5}
\multicolumn{1}{ |c  }{5200} &
\multicolumn{1}{ |c| }{0.0381} & 1.0719 & 0.0365 & 1.0177 & \multicolumn{1}{ c| }{} &\multicolumn{1}{ c|  }{} \\ \cline{1-5}
\multicolumn{1}{ |c  }{5300} &
\multicolumn{1}{ |c| }{0.0381} & 1.0735 & 0.0365 & 1.0170 & \multicolumn{1}{ c| }{} &\multicolumn{1}{ c|  }{} \\ \cline{1-5}
\multicolumn{1}{ |c  }{5400} &
\multicolumn{1}{ |c| }{0.0382} &1.0735 & 0.0365 & 1.0169 & \multicolumn{1}{ c| }{} &\multicolumn{1}{ c|  }{} \\ \cline{1-5}
\multicolumn{1}{ |c  }{5500} &
\multicolumn{1}{ |c| }{0.0381} &1.0720 & 0.0365 & 1.0174 & \multicolumn{1}{ c| }{} &\multicolumn{1}{ c|  }{} \\ \cline{1-5}
\multicolumn{1}{ |c  }{5600} &
\multicolumn{1}{ |c| }{0.0381} &1.0730 & 0.0365 & 1.0175 & \multicolumn{1}{ c| }{} &\multicolumn{1}{ c|  }{} \\ \cline{1-5}
\multicolumn{1}{ |c  }{5700} &
\multicolumn{1}{ |c| }{0.0381} & 1.0737 & 0.0365 & 1.0169 & \multicolumn{1}{ c| }{} &\multicolumn{1}{ c|  }{} \\ \cline{1-5}
\multicolumn{1}{ |c  }{5800} &
\multicolumn{1}{ |c| }{0.0381} & 1.0725 & 0.0365 & 1.0170 & \multicolumn{1}{ c| }{} &\multicolumn{1}{ c|  }{} \\ \cline{1-5}
\multicolumn{1}{ |c  }{5900} &
\multicolumn{1}{ |c| }{0.0381} &1.0723 & 0.0365 & 1.0177 & \multicolumn{1}{ c| }{} &\multicolumn{1}{ c|  }{} \\ \cline{1-7}
\end{tabular}
\caption{Measured speed and amplitude of the leading solitary wave over long time scales, $\varepsilon = 0.05$.}
\label{T:1}
}
\end{center}
\end{table}

We now present the results related to time evolution of scattering data associated to the solutions in the numerical experiments discussed above. Given an eigenvalue $\lambda = \frac{1}{2} \left(\zeta + \zeta^{-1} \right)$ of $L$, and the associated norming constant $\gamma$, the corresponding Toda soliton is given by
\begin{equation}
a_n(t) = \frac12\frac{\sqrt{1 - \zeta^2 + \gamma(t)\zeta^{2(n-1)}}\sqrt{1 - \zeta^2 + \gamma(t)\zeta^{2(n+1)}}}{1 - \zeta^2 + \gamma(t)\zeta^{2n}}
\end{equation}
via the inverse scattering transform \cite{KT_rev, KT_sol}. The amplitude of this wave is given by
\begin{equation*}
\sup_{n\in\mathbb{Z}} \Bigg(a_n (0) - \frac 12 \Bigg) = \frac{|\lambda | - 1}{2} = \frac{(1 - |\zeta | )^2}{4|\zeta |}\,,
\end{equation*}
as in \cite{IST}. Therefore the eigenvalues associated to taller Toda solitons are located farther away from the continuous spectrum. Figure~\ref{F:SolEV_x2x3} shows that the eigenvalue, which initially corresponds to a Toda soliton, rapidly converges to a new asymptotic constant that is closer to the continuous spectrum. This behavior is consistent with the loss in amplitudes of these waves. Moreover, the influence of the cubic perturbation, which yields the smaller of the emerging solitary waves, accordingly drives the associated eigenvalue closer to the edge of the ac spectrum compared to the case with the quadratic perturbation. Similar behavior is observed for different choices of perturbations or initial data, including one with multiple eigenvalues.

\begin{figure}[h!]\centering
\includegraphics[width=255pt]{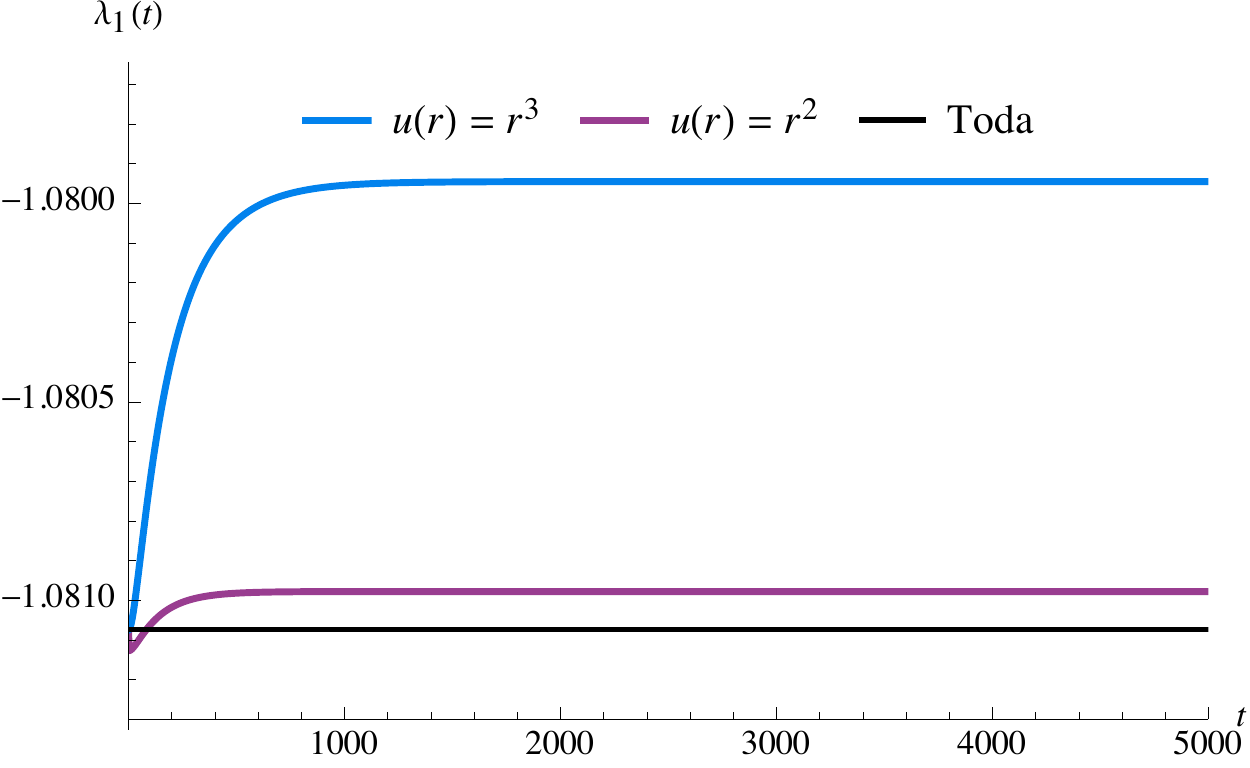}
\caption{ Evolution of the eigenvalue associated with 1-soliton initial data for different perturbations, \mbox{$\varepsilon = 0.05$}.}
\label{F:SolEV_x2x3}
\end{figure}
Figure~\ref{F:EV_multi_sol} displays time evolution of eigenvalues in a numerical experiment where we place two equal sized Toda solitons far away from each other in the spatial domain of integration, and let them evolve towards each other in the perturbed lattice with $u(r)=r^2$. 
\begin{figure}[h]\centering
\subfigure[]{\includegraphics[width=188pt]{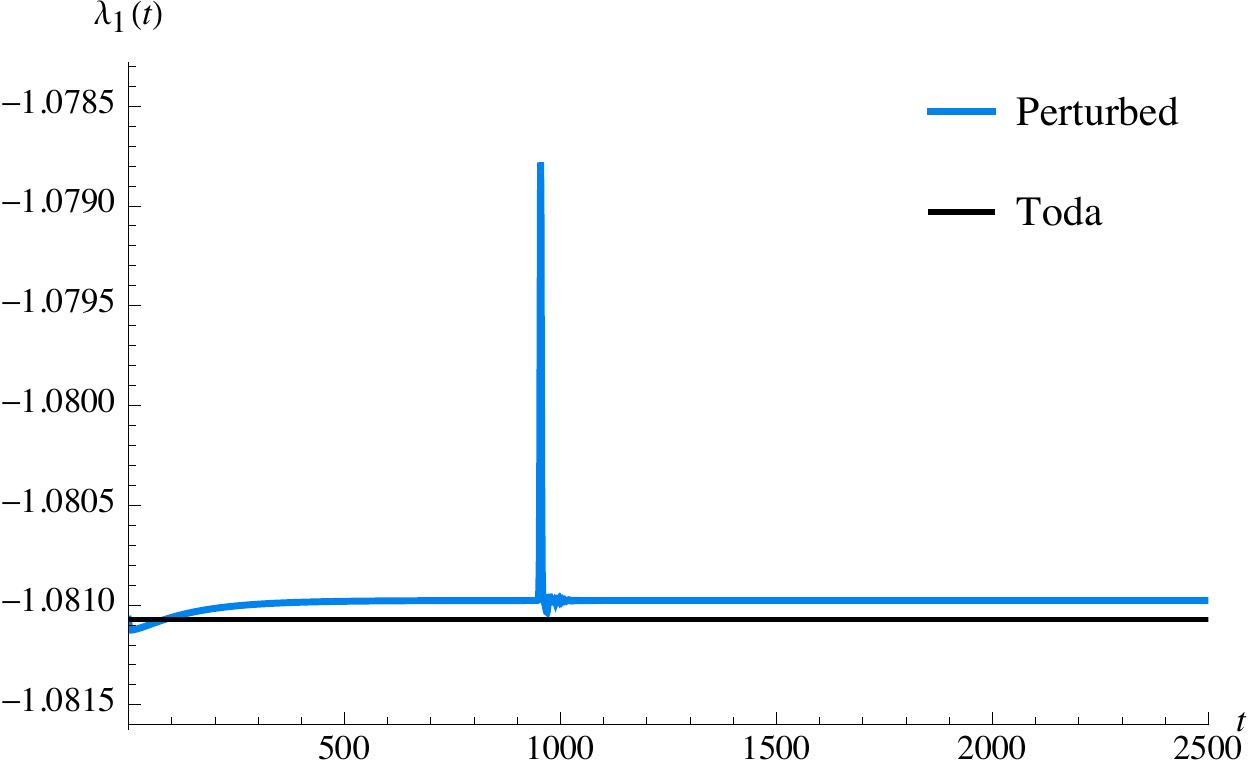}}\quad
\subfigure[]{\includegraphics[width=188pt]{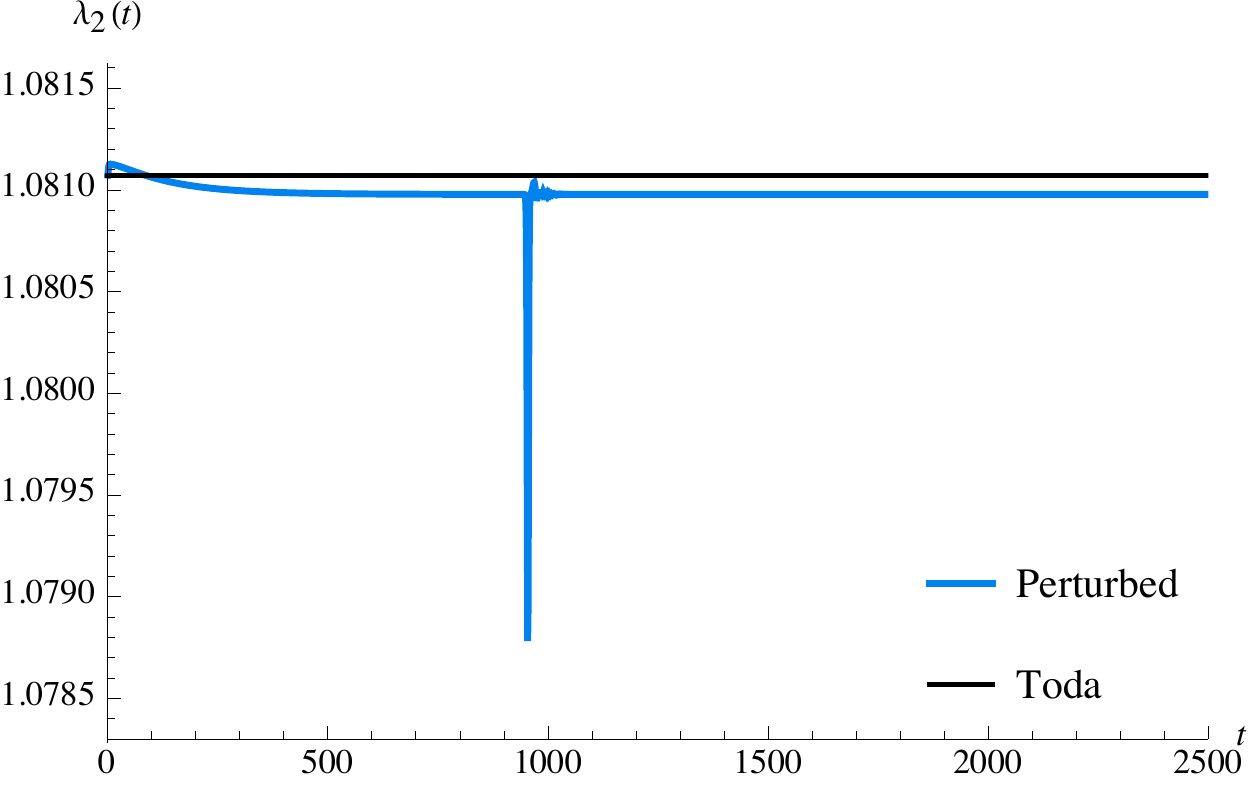}}
\caption{Eigenvalues for  initial data with two counter-propagating 1-solitons under the perturbed lattice (blue) vs. the Toda lattice (black), $\varepsilon = 0.05$, $u(r)=r^2$.}
\label{F:EV_multi_sol}
\end{figure}

As can be seen, both of the eigenvalues (blue) converge to new asymptotic values. Figure~\ref{F:EV_multi_sol}(a) displays the evolution of the eigenvalue for the left-incident wave plotted against the evolution in absence of perturbation (black), and Figure~\ref{F:EV_multi_sol}(b) is the analogous picture for the right-incident wave. The peaks and small oscillations formed over time in the trajectories of these eigenvalues are due to the interactions of the solitary waves with each other, and with the radiation that has been generated. We address the case of interacting solitary waves in greater detail in the next section, but it is worthwhile to note that the eigenvalues revert to their asymptotic values after the interactions.

Next, we study propagation speeds of the emerging solitary waves and their relation to time evolution of the corresponding scattering data. The numerical investigation that is to be discussed below is similar to the direction pursued in~\cite{FFM}. Note that the speed of a Toda soliton (corresponding to spectral parameters $(\zeta_j, \gamma_j)$) in terms of the scattering data is given by
\begin{equation}\label{E:speed}
v_j = - \frac{\theta(\zeta_j)}{\log(|\zeta_j|)}\,,
\end{equation}
where $\theta(\zeta_j) = \tfrac{1}{2}(\zeta_j - \zeta_j^{-1})$ is defined via $\gamma_j(t) = \gamma_j^{(0)} e^{2\theta(\zeta_j)t}$, which describes the evolution of the norming constant when $\varepsilon=0$, with $\gamma_j^{(0)} = \gamma_j(0)$. We consider the perturbed lattices with $u(r)=r^2$ and $u(r)=r^3$, for varying values of $\varepsilon$ and 1-soliton initial data. Through numerically computing the evolution of the spectral parameter $\zeta_1$ and the associated norming constant $\gamma_1$, we find that as $t \rightarrow \infty$
\begin{equation}\label{E:eig_nc_as}
\begin{aligned}
\zeta_1(t) &\sim \zeta_{1,\varepsilon} \,,\\
\log \gamma_1(t) &\sim \log \gamma_{1,\varepsilon} + \omega_{1,\varepsilon} t\,,
\end{aligned}
\end{equation}
where $\zeta_{1,\varepsilon}$, $\gamma_{1,\varepsilon}$, and $\omega_{1,\varepsilon}$ are constants that depend on the initial data and the perturbation. In Figure~\ref{F:zeta}, we plot the evolution of the spectral parameter $\zeta_1$ from time $t=0$ until $t=2500$, in the same numerical experiments. Clearly, the amount of deviation of $\zeta_1(t)$ from the initial value becomes larger as the perturbation size increases.
\begin{figure}[htp]\centering
\subfigure[$u(r)=r^2$]{\includegraphics[width=188pt]{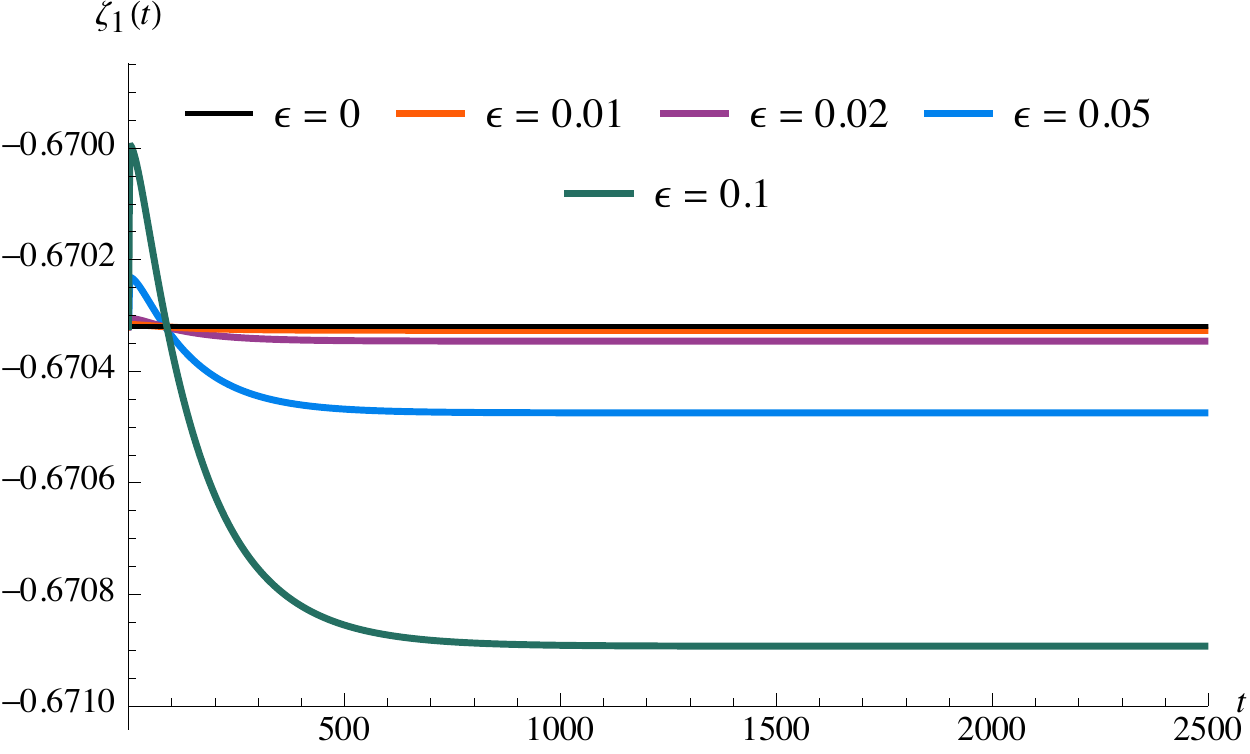}}\quad
\subfigure[$u(r)=r^3$]{\includegraphics[width=188pt]{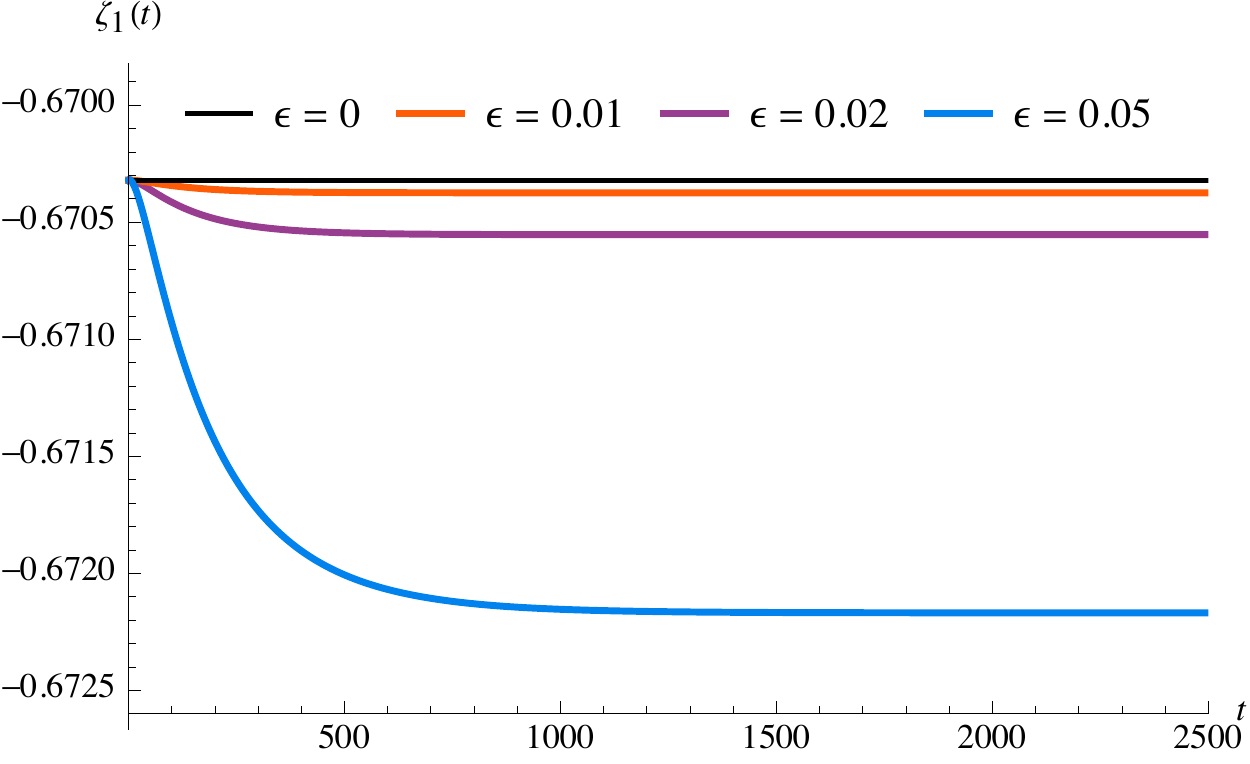}}
\caption{Spectral parameter $\zeta_1$ corresponding to \mbox{1-soliton} initial data for different choices of perturbations $u$.}
\label{F:zeta}	
\end{figure}

Figure~\ref{F:norming_constants} shows the evolution of $\log\gamma_1(t)$, the logarithm of the norming constant, from time $t=0$ until $t=6000$ in the same experiments. In Figure~\ref{F:norming_constants}(a)-(b) we see that $\log\gamma_1 (t)$ is asymptotically linear for large values of $t$, as described in \eqref{E:eig_nc_as}. Figure~\ref{F:norming_constants}(c) and (d) provides a closer look at the evolution of $\log\gamma_1$ under perturbations $u(r)=r^2$ and $u(r)=r^3$, with different values of perturbation size $\varepsilon$. 
\begin{figure}[h!]\centering
\subfigure[$u(r)=r^2$]{\includegraphics[width=188pt]{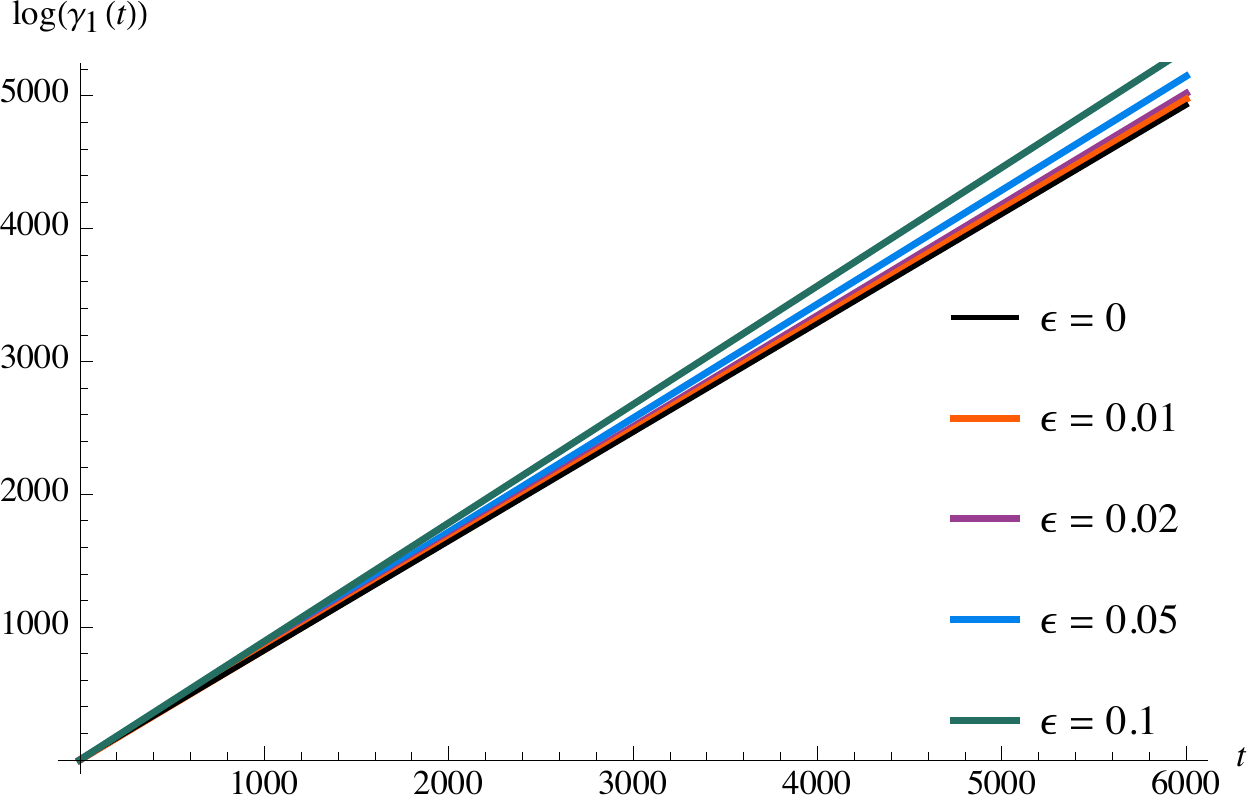}}\quad
\subfigure[$u(r)=r^3$]{\includegraphics[width=188pt]{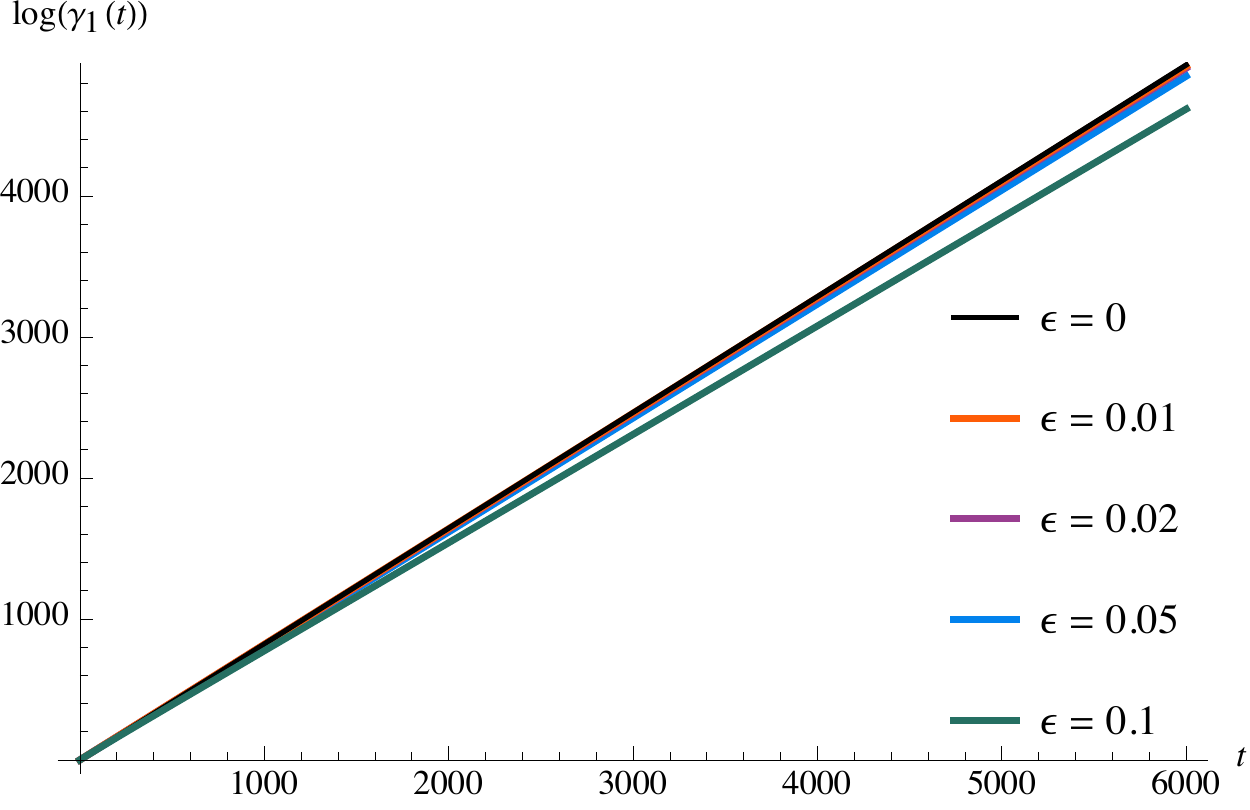}}\\
\subfigure[$u(r)=r^2$]{\includegraphics[width=188pt]{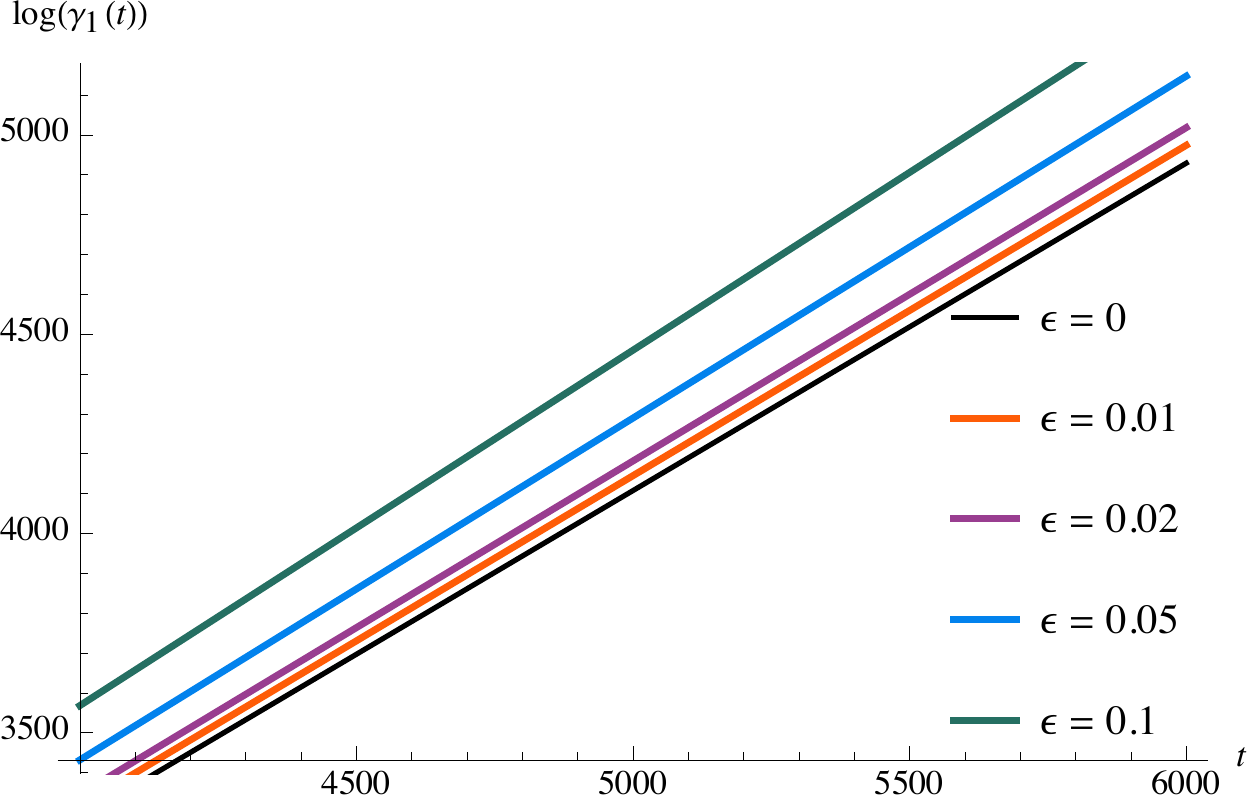}}\quad
\subfigure[$u(r)=r^3$]{\includegraphics[width=188pt]{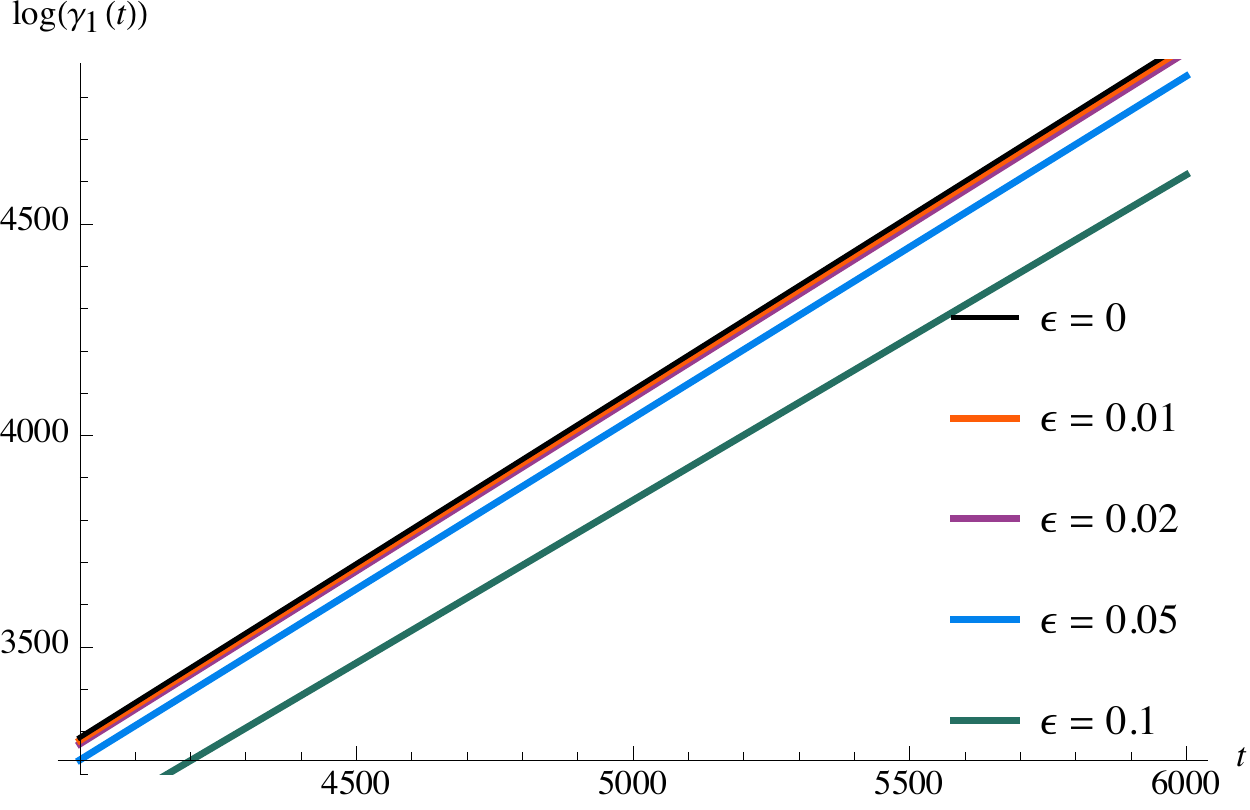}}
\caption{Logarithm of the norming constant for different choices of perturbation functions $u$, and $\varepsilon$.}
\label{F:norming_constants}	
\end{figure}
As $\varepsilon$ increases, $\log \gamma_1$ drifts farther away from the trajectory it has in case $\varepsilon=0$. In other words, $\omega_{1,\varepsilon}$ in \eqref{E:eig_nc_as} changes monotonically with respect to $\varepsilon$: increasing for $u(r)=r^2$ and decreasing for $u(r)=r^3$. Using \eqref{E:speed} and \eqref{E:eig_nc_as}, we find that the measured speed of the leading solitary wave in the perturbed lattices is asymptotically given by
\begin{equation}\label{E:as_speed}
v_1(\omega_{1,\varepsilon}, \zeta_{1,\varepsilon}) \sim - \frac{\omega_{1,\varepsilon}}{2\log(|\zeta_{1,\varepsilon}|)}\,,
\end{equation}
which is analogous to the expression \eqref{E:speed}. Finally, we note that \eqref{E:eig_nc_as} implies $\partial_t \| \varphi_{+}(\zeta_1; t) \|^2_{\ell^2} \sim - \omega_{1,\varepsilon} \| \varphi_{+}(\zeta_1; t) \|^2_{\ell^2}$, for the Jost solution $\varphi_{+}$, as $t\rightarrow +\infty$. 

\subsection{Emergence of new eigenvalues}\label{S:soliton_free_data}
In all of the numerical experiments we consider with 1-soliton initial data, we find that new eigenvalues are pushed out of the continuous spectrum after a small amount of time elapses. By the Remark~\ref{R:new_eigs}, presence of new eigenvalues in the spectrum of the truncated operator implies existence of new eigenvalues in the spectrum of $L$.
\begin{figure}[htp]\centering
\subfigure[$u(r)=r^2$]{\includegraphics[width=188pt]{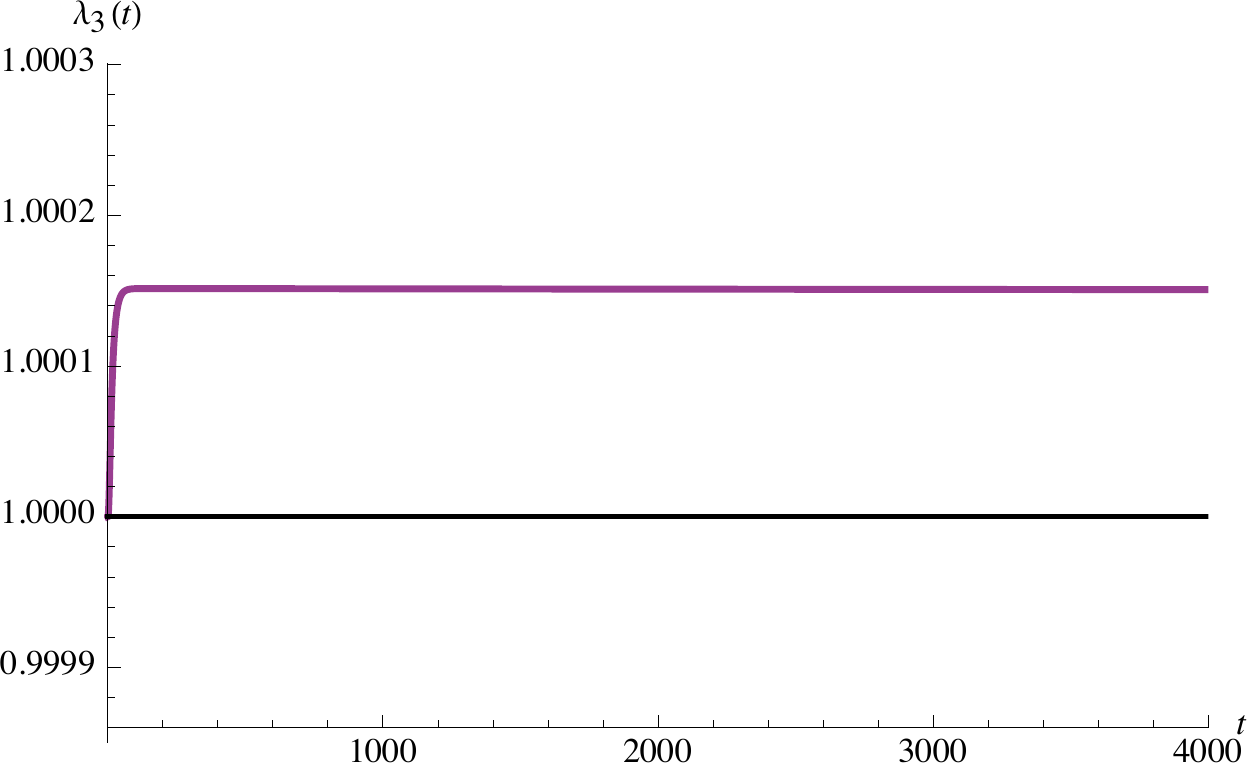}}\quad
\subfigure[$u(r)=r^3$]{\includegraphics[width=188pt]{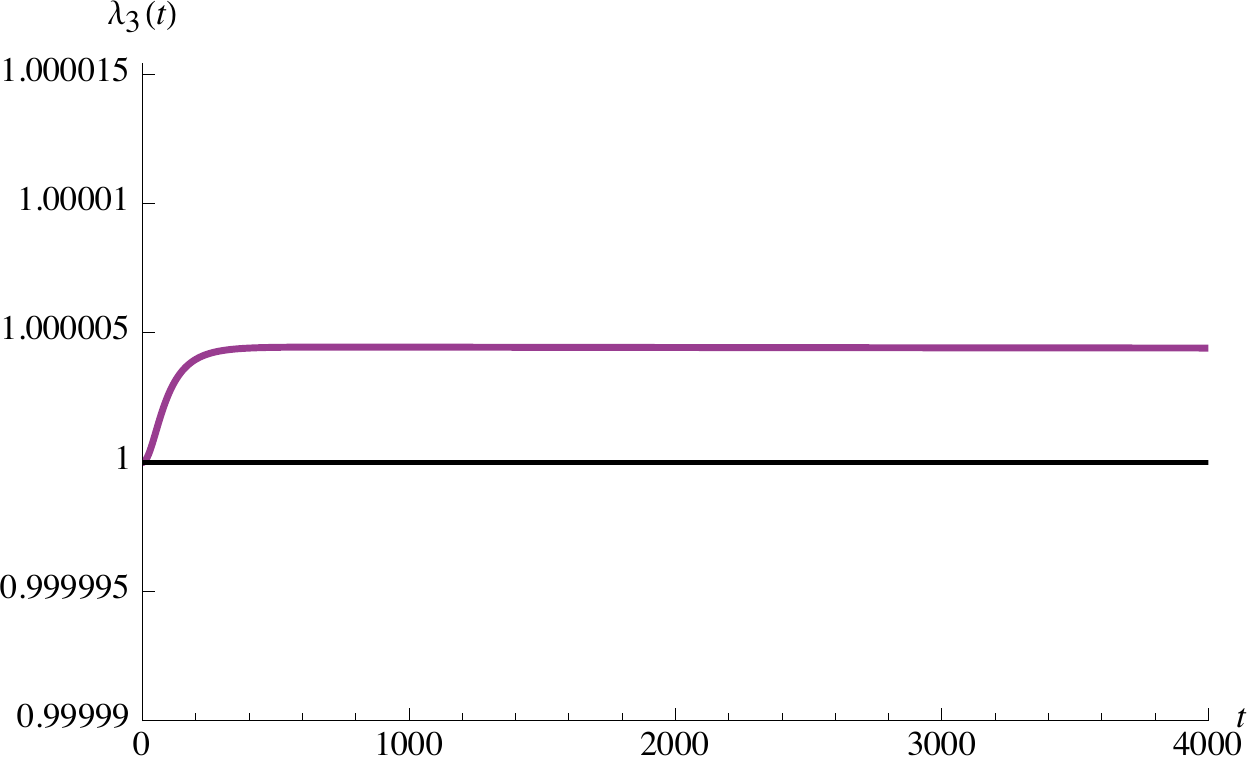}}
\caption{New eigenvalues emerging from the opposite side of the ac spectrum of $L$, $\varepsilon = 0.05$}
\label{F:newEV_opposite}	
\end{figure}

In numerical experiments that are mentioned in this section, the eigenvalue corresponding to the 1-soliton initial data lies in $(-\infty, -1)$. Figure~\ref{F:newEV_opposite} displays time evolution of new eigenvalues which are being pushed into the opposite side of the continuous spectrum, namely into $(1, \infty)$ for $u(r)=r^2$ in (a), and $u(r)=r^3$ in (b), 
with $\varepsilon = 0.05$. Long-time behavior of these new eigenvalues are in general not clear in the time scales are able to 
conduct our experiments.

Moreover, we also find that there are new eigenvalues being pushed into $(-\infty, -1)$, the side of the continuous spectrum that contains the eigenvalue associated to the initial data. Figure~\ref{F:newEV_same} displays the evolution of these new eigenvalues coming into the negative side of the ac spectrum in the same numerical experiment, with $u(r)=r^2$ in (b), and $u(r)=r^3$ in (a).
\begin{figure}[h]\centering
\subfigure[$u(r)=r^2$]{\includegraphics[width=188pt]{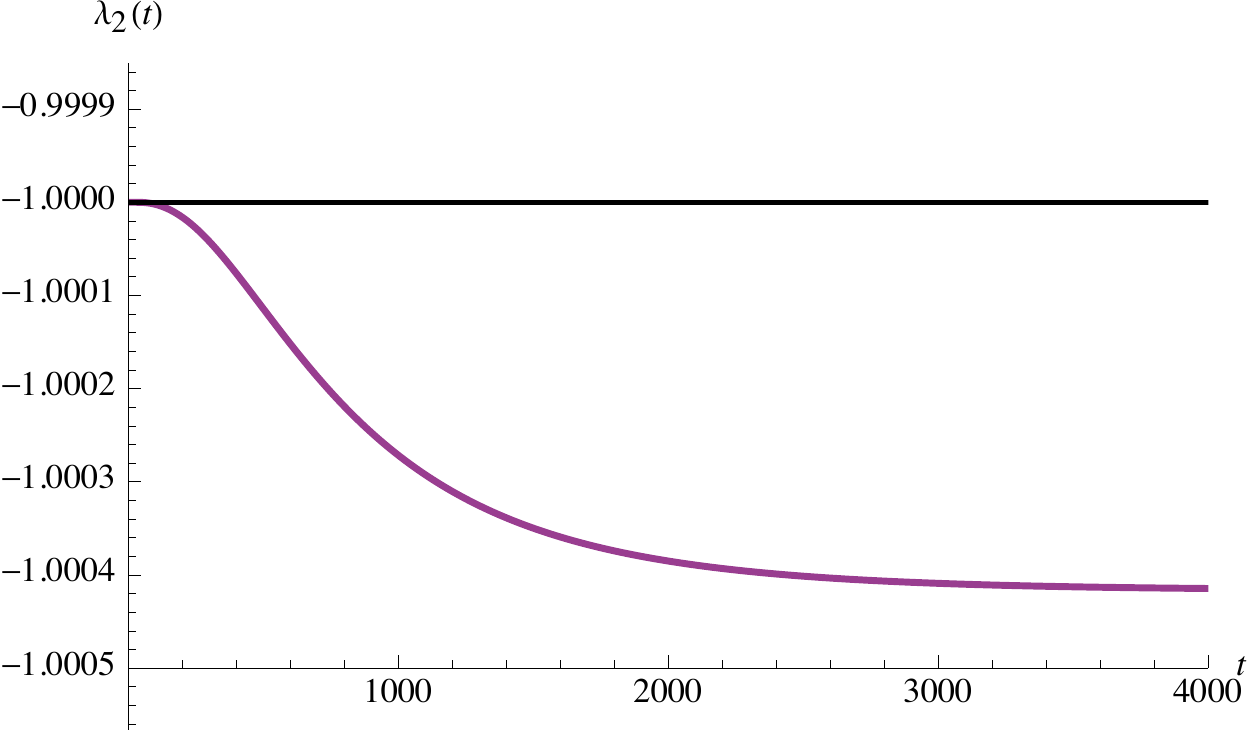}}\quad
\subfigure[$u(r)=r^3$]{\includegraphics[width=188pt]{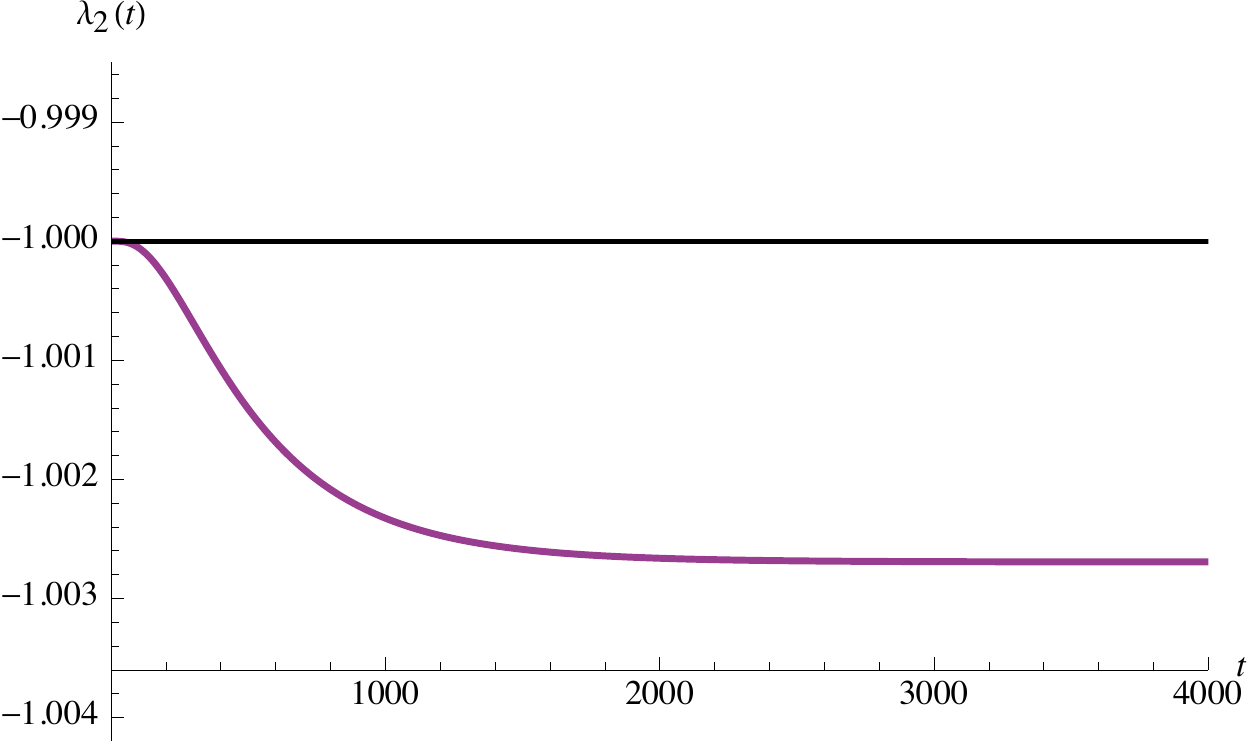}}
\caption{New eigenvalues emerging from the same side of the ac spectrum.}
\label{F:newEV_same}	
\end{figure}

Note that the spatial decay of the solutions $(a,b)$ given in Theorem~\ref{T:SpatialDecay} rules out the possibility of having embedded eigenvalues at the edges $\lambda=\pm 1$ of the essential spectrum. Therefore, $L$ does not have 
eigenvalues at $\lambda=\pm 1$ at $t=0$. What we see here is that the resonances of the Jacobi matrix at $\lambda=\pm 1$ are being pulled out by the perturbed dynamics as time evolves, see, for example,~\cite{Sim}. Emergence of new eigenvalues in the spectrum of $L$ therefore implies that the scattering data associated to $L$ at a time $t^* >0$ only partially coincide with the scattering data obtained through solving the evolution equations in Theorem~\ref{T:eom_p_eigs} (in~\ref{A:app}) from time $t=0$ to $t=t^*$.

\section{Numerical Results: Absence of Eigenvalues in the Initial Data}

Now we consider the solutions of \eqref{E:eom_p_ab} with initial data $(a^{0}, b^{0})$ such that the discrete spectrum of $L$ is initially empty. In the Toda lattice, such initial data yields a purely dispersive solution $(a,b)$, i.e.\
\begin{equation*}
\lim_{t\rightarrow \infty} \big\| a(t) -\tfrac{1}{2} \big\|_{\ell^{\infty}} + \big\| b(t) \big\|_{\ell^{\infty}} = 0\,.
\end{equation*}
One way to construct such initial data is as follows. Note that if $0 < a_n < \frac{1}{2}$ and $b_n=0$ for all $n\in\mathbb{Z}$, then the quadratic form $Q_L$ associated to the doubly-infinite Jacobi matrix $L$ satisfies $-1 < Q_L (\phi)< 1$, for all $\phi$ with $\|\phi\|_{\ell^2(\mathbb{Z})}=1$. This immediately implies that $L$ has no eigenvalues. In the numerical experiment to be discussed now, we consider a Toda soliton $\big(\hat{a}, \hat{b}\big)$ with $\big\|\hat{a} - \tfrac{1}{2}\big\|_{\ell^{\infty}} < \tfrac{1}{5}$, and commence with the initial data obtained through setting
\begin{equation}\label{E:no_eig_ID}
a^{0}_n = -(\hat{a}_n -\tfrac{1}{2}) + \tfrac{1}{2} = 1 - \hat{a}_n~\text{~and~}~ b^{0}_n = 0\,,
\end{equation}
which corresponds to reflecting the solitary wave profile $\hat{a}$ vertically and setting the initial velocity of each particle equal to $0$. Since \mbox{$\tfrac{1}{2} <\hat{a}_n <\tfrac{3}{5} $} for all $n\in\mathbb{Z}$, we have \mbox{$0<a^0_n < \tfrac{1}{2}$}, and hence the Jacobi matrix corresponding to the initial data $(a^0, b^0)$ has an empty discrete spectrum. Note that this transformation preserves all of the required spatial decay conditions given in Theorem~\ref{T:SpatialDecay} (in~\ref{A:app}).
\begin{figure}[h!]\centering
\subfigure[]{\includegraphics[width=188pt]{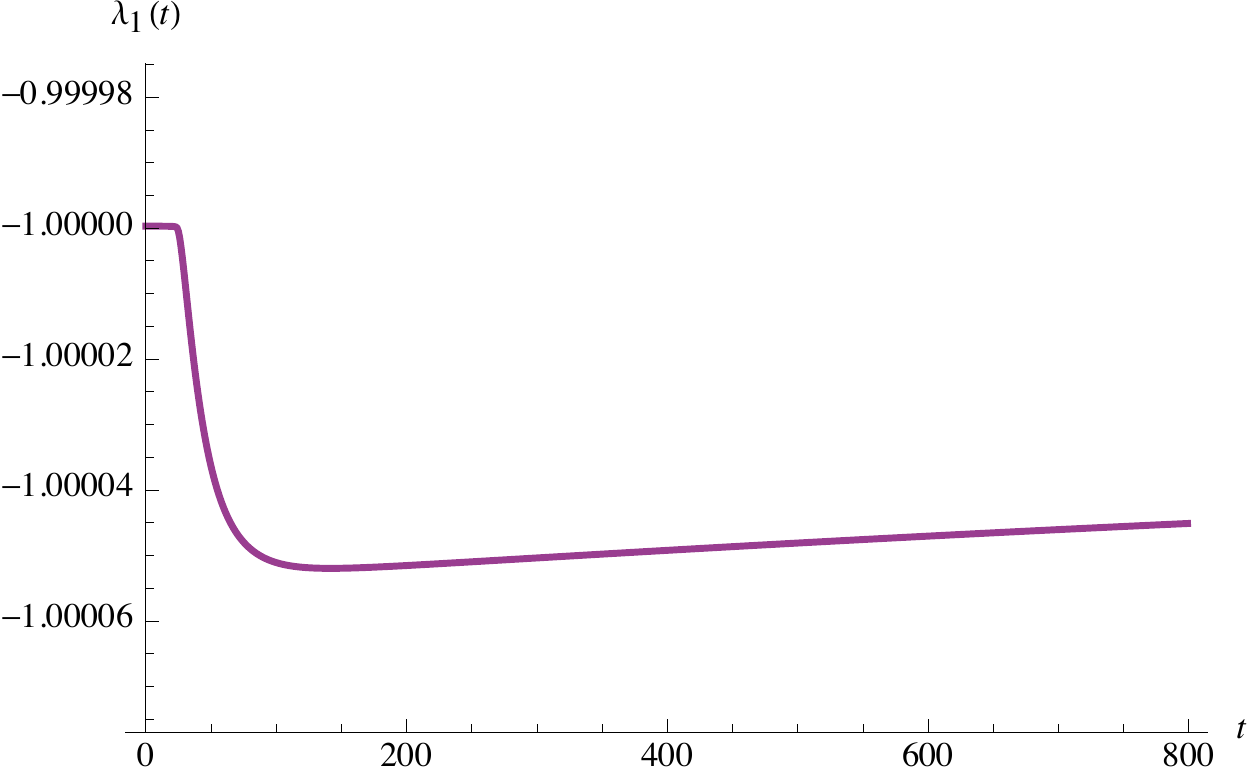}}\quad
\subfigure[]{\includegraphics[width=188pt]{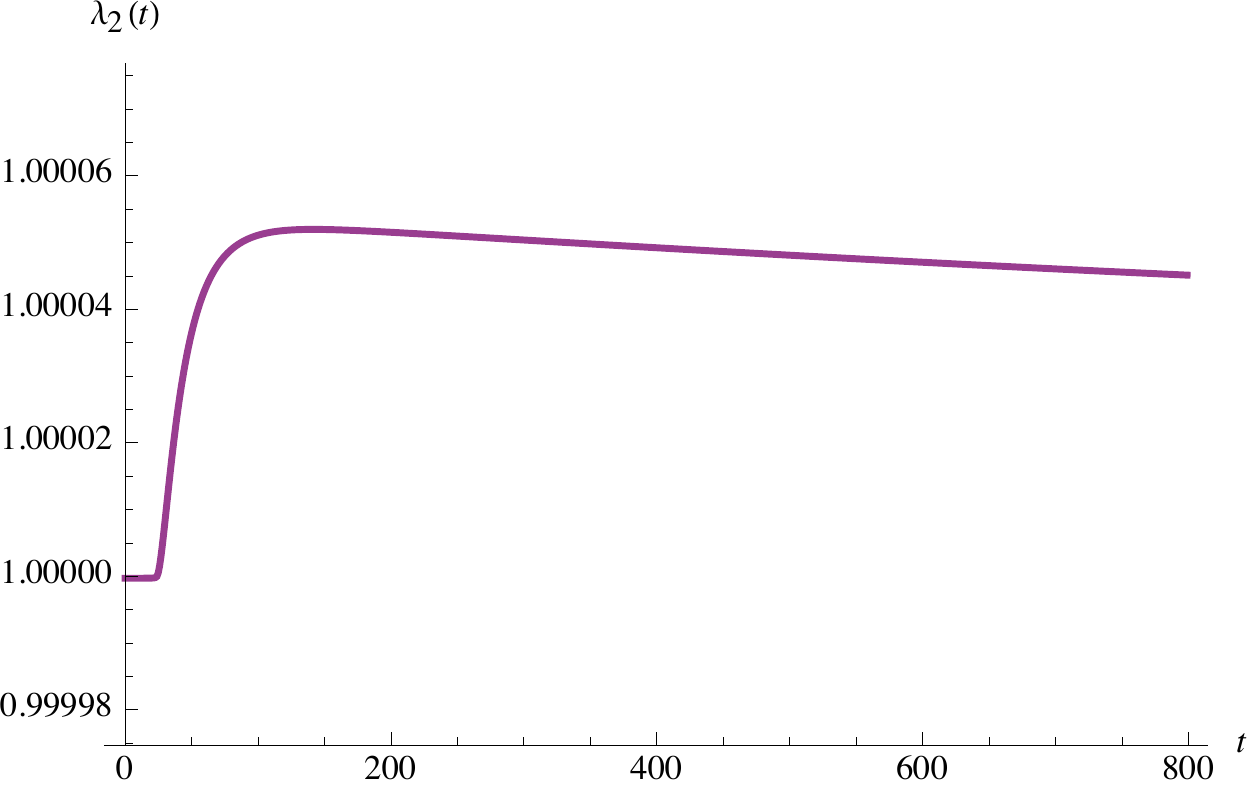}}
\caption{New eigenvalues emerging from both sides of the ac spectrum, $u(r)=r^2$, $\varepsilon = 0.05$, $N=2^{12}$; $\Delta t=10^{-3}$.}
\label{F:noeig1_neweig_pert}
\end{figure}
Although we start with an empty discrete spectrum, we find that new eigenvalues emerge from both hand sides of the continuous spectrum under the perturbed dynamics. Figure~\ref{F:noeig1_neweig_pert} displays the evolution of two new eigenvalues that emerge from the opposite sides of the ac spectrum. In any case, the new eigenvalues do not seem to converge to asymptotical values that are outside $[-1,1]$ in the time-scale of the numerical experiments.

\section{Numerical Results: Clean Solitary Waves}\label{S:clean}
We now study interactions of solitary waves in the perturbed lattices. To pursue such a study, one needs to have accurate numerical approximations of these solitary wave solutions since an exact formula for such solutions is not available at hand. We are able to generate ``clean" solitary waves numerically through the iterative procedure introduced in \cite{BC}. We commence with initial data that is a pure 1-soliton solution of the Toda lattice and let it evolve numerically under the perturbed dynamics for a relatively long time. Once a leading solitary wave is separated from the dispersive tail, and from the smaller (hence slower) solitary waves that may emerge, we cut it off by setting the remainder of the solution equal to the free solution $\left(a \equiv \tfrac{1}{2}, b\equiv 0 \right)$. Then we place the numerically isolated wave in the middle of the spatial domain of numerical integration and repeat this process. We note that more than one iterations of this procedure were needed in order to obtain an accurate solitary wave.

\subsection{Head-on collision of solitary waves}
Once we have in hand a good numerical approximation of a solitary wave solution of the perturbed system under study, we set up numerical experiments in which a pair of identical clean solitary waves travel towards each other, and interact in a head-on collision. We study the time evolution of the collision, as well as the time evolution of the eigenvalues (of $L$) associated to these solitary waves in the same time window. As is well known, Toda solitons exhibit elastic collision, i.e.\ they retain their shapes and speeds after the collision, and no radiation is produced during their interaction. The situation is slightly different for clean solitary waves in the perturbed lattices, as we shall see in the numerical experiment that is to be discussed now.

Figure \ref{F:SolWaveCol} displays two equal-sized solitary waves traveling towards each other in the perturbed lattice with the perturbation \mbox{$u(r)=r^2$} and \mbox{$\varepsilon=0.05$}, in a time window where the interaction occurs. 
\begin{figure}[h!]\centering
\subfigure[$t=1870$]{\includegraphics[width=188pt]{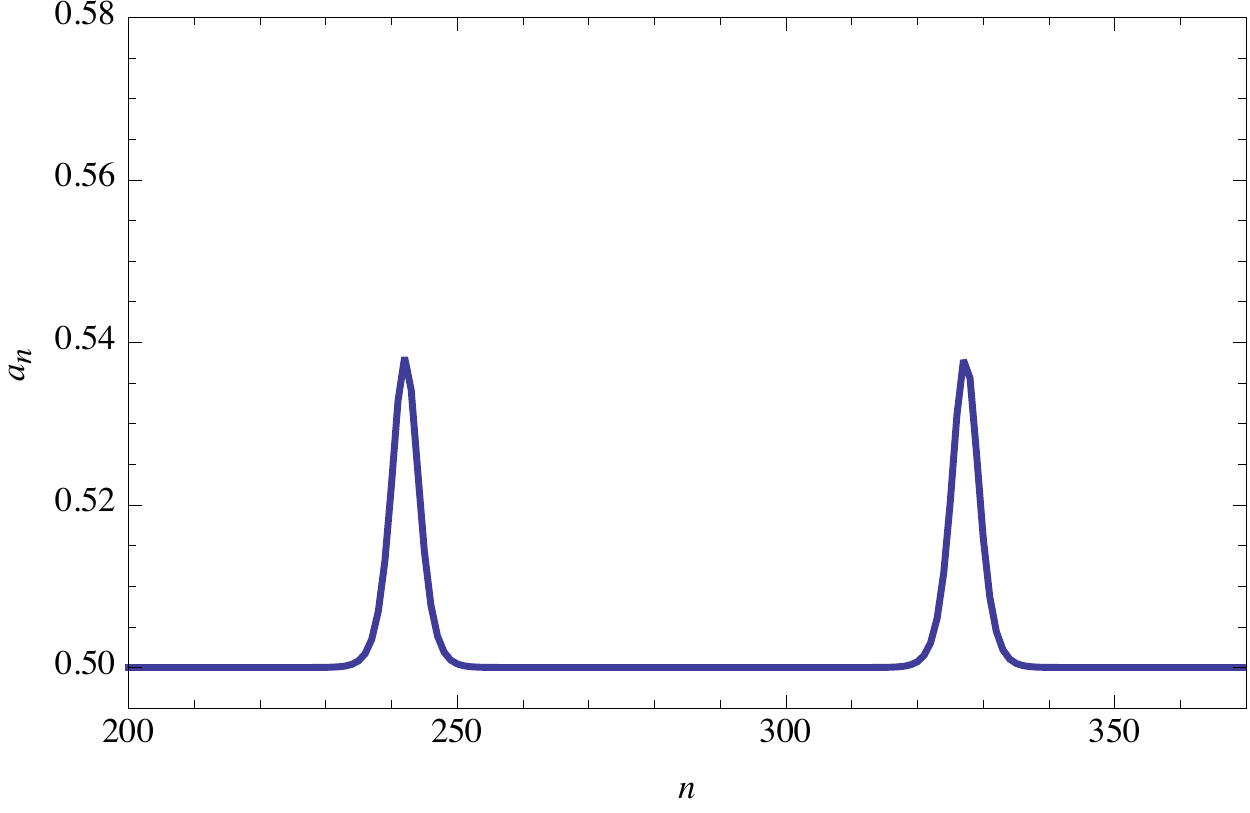}}\quad
\subfigure[$t=1905$]{\includegraphics[width=188pt]{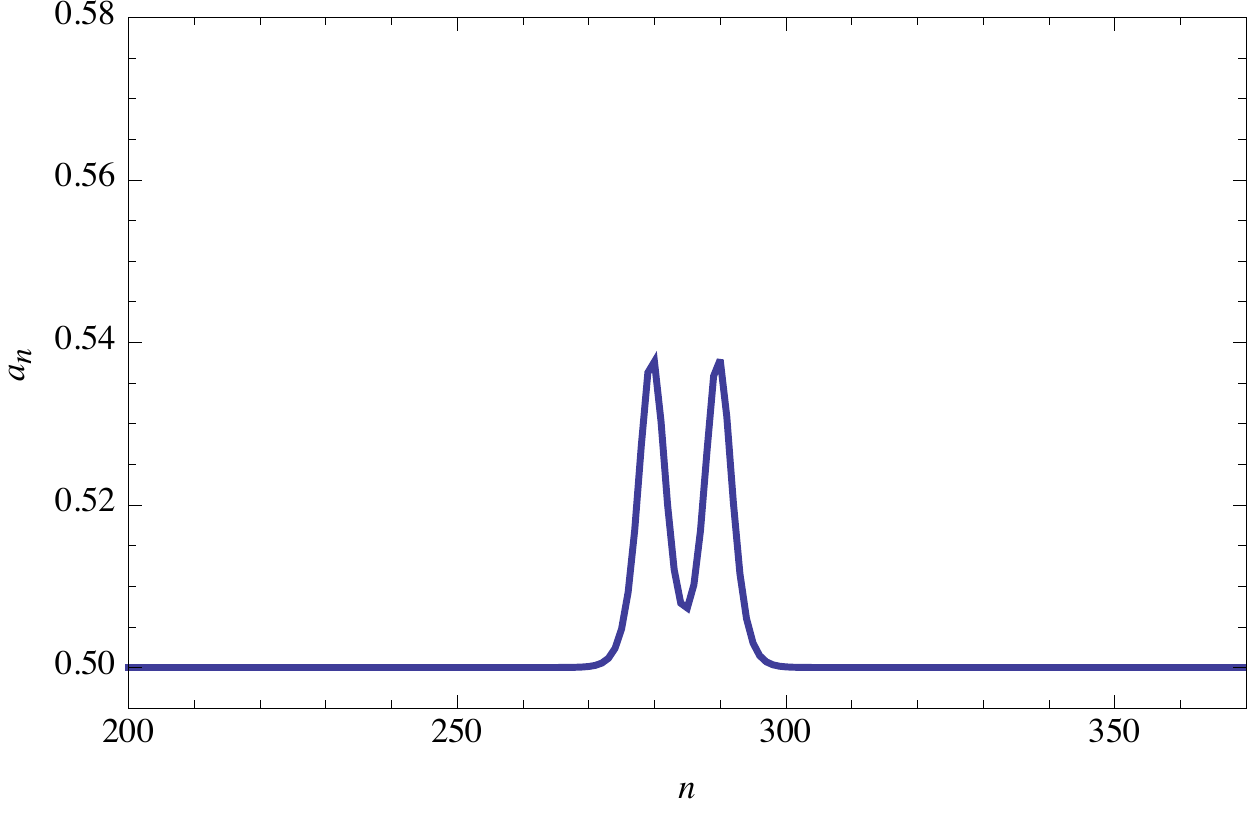}}\\
\subfigure[$t=1910$]{\includegraphics[width=188pt]{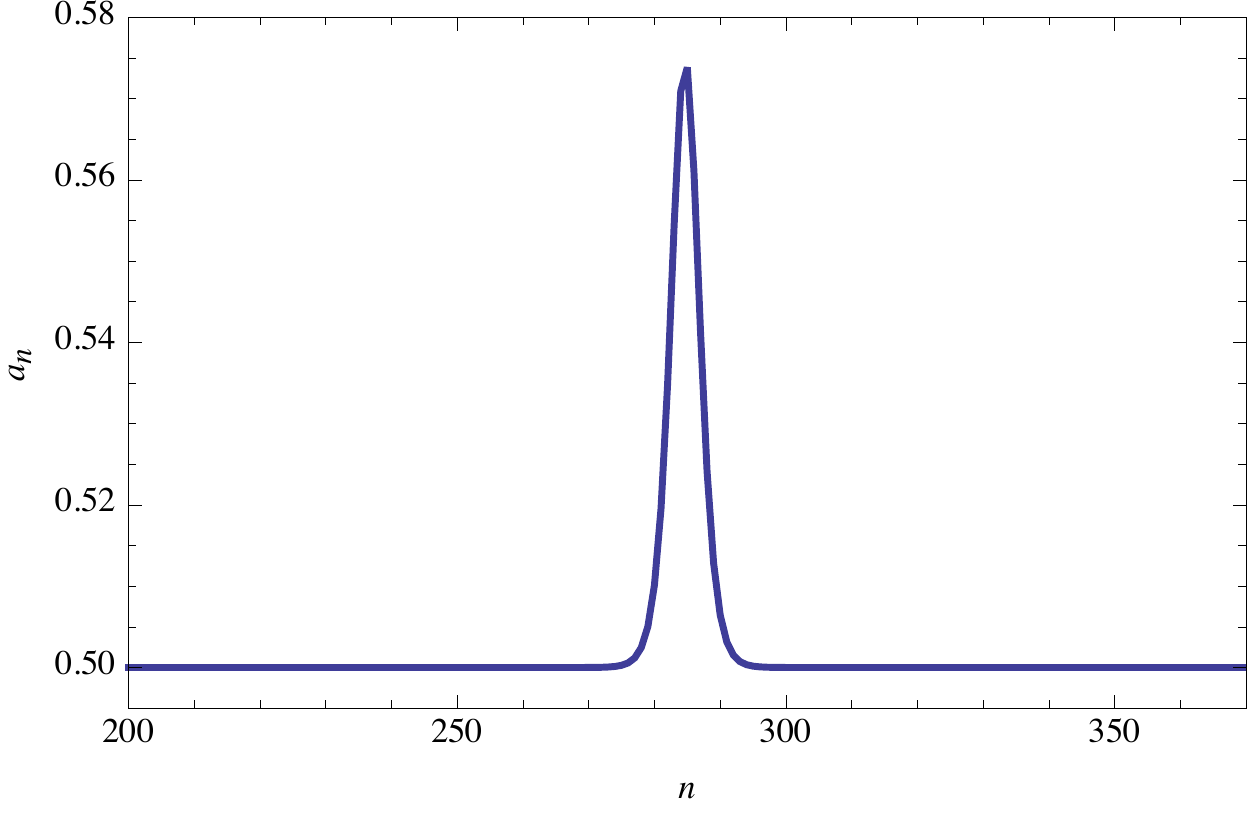}}\quad
\subfigure[$t=1945$]{\includegraphics[width=188pt]{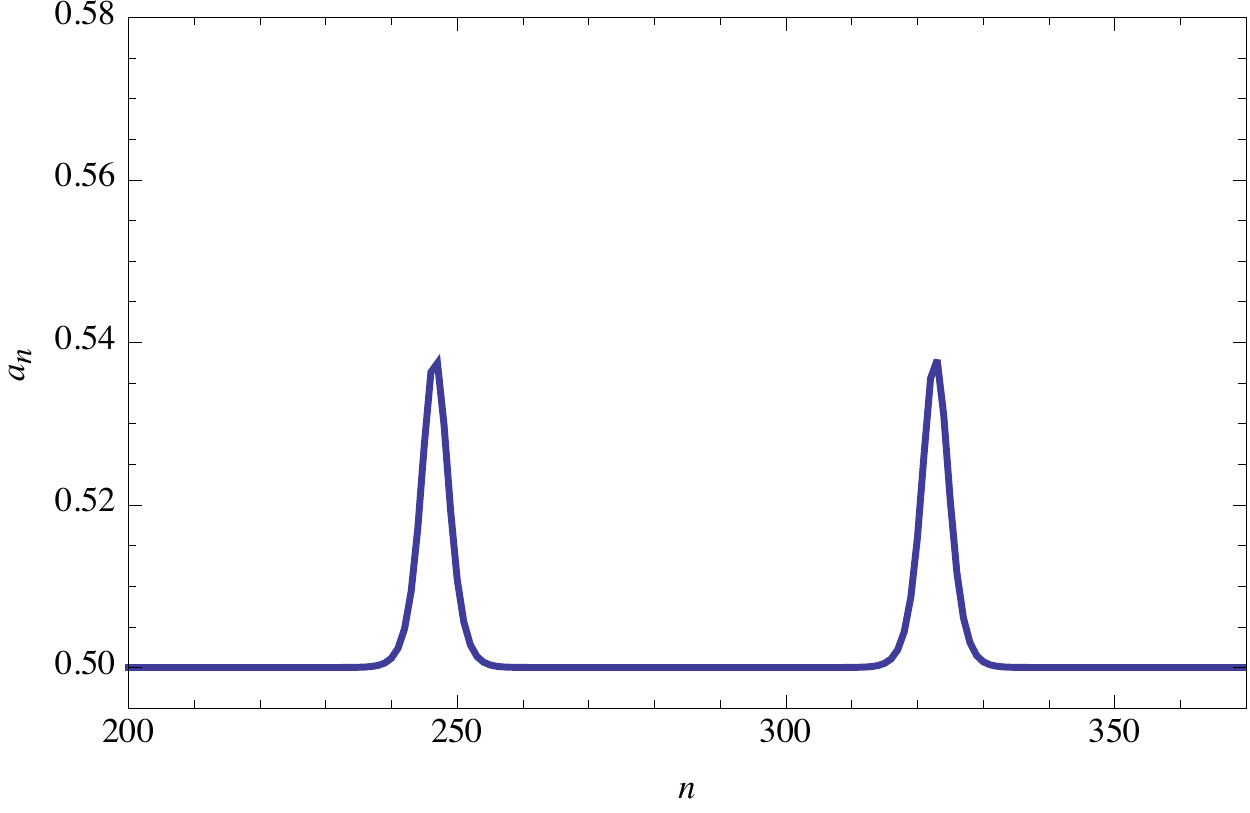}}
\caption{Interaction of two equal-sized clean solitary waves during head-on collision. $u(r)=r^2$, $\varepsilon=0.05$, $\Delta t =10^{-3}$.}
\label{F:SolWaveCol}
\end{figure}
From Figure~\ref{F:SolWaveCol}(a) to (b) the waves propagate towards each other from time $t=1870$ until $t=1905$, where they are already in the collision state. At time $t=1910$, the interacting profile forms a peak as seen in Figure~\ref{F:SolWaveCol}(c), and Figure~\ref{F:SolWaveCol}(d) displays the wave profiles at $t=1945$, separating from each other.
A closer look at the (spatial) interval between the separating waves in Figure~\ref{F:SolWaveCol}(d) reveals that there develops some radiation with very small amplitude.

Figure~\ref{F:SolWaveCol_closer}(a) displays the solitary waves before the interaction, and Figure~\ref{F:SolWaveCol_closer}(b) shows the radiation development between the separating solitary waves. Amplitude of this wave decays with time as is usual for purely dispersive waves. 
\begin{figure}[h!]\centering
\subfigure[$t=1870$]{\includegraphics[width=188pt]{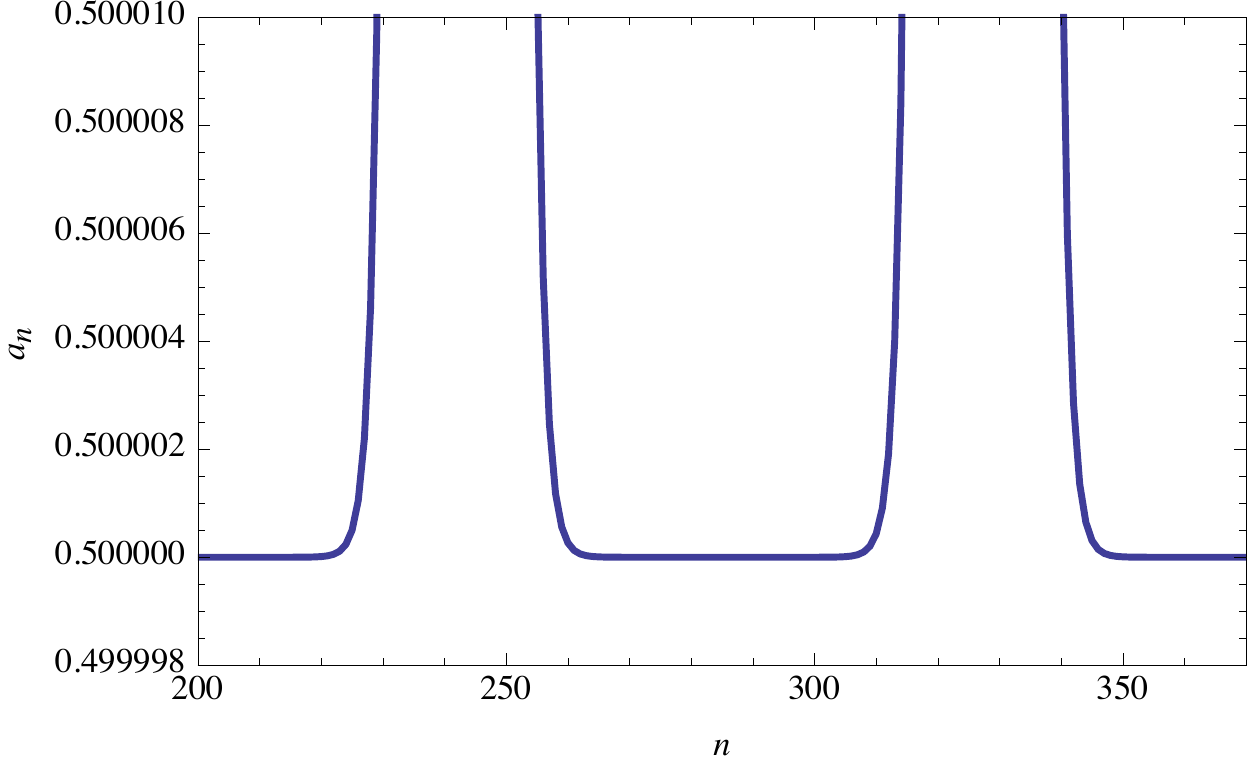}}\quad
\subfigure[$t=1950$]{\includegraphics[width=188pt]{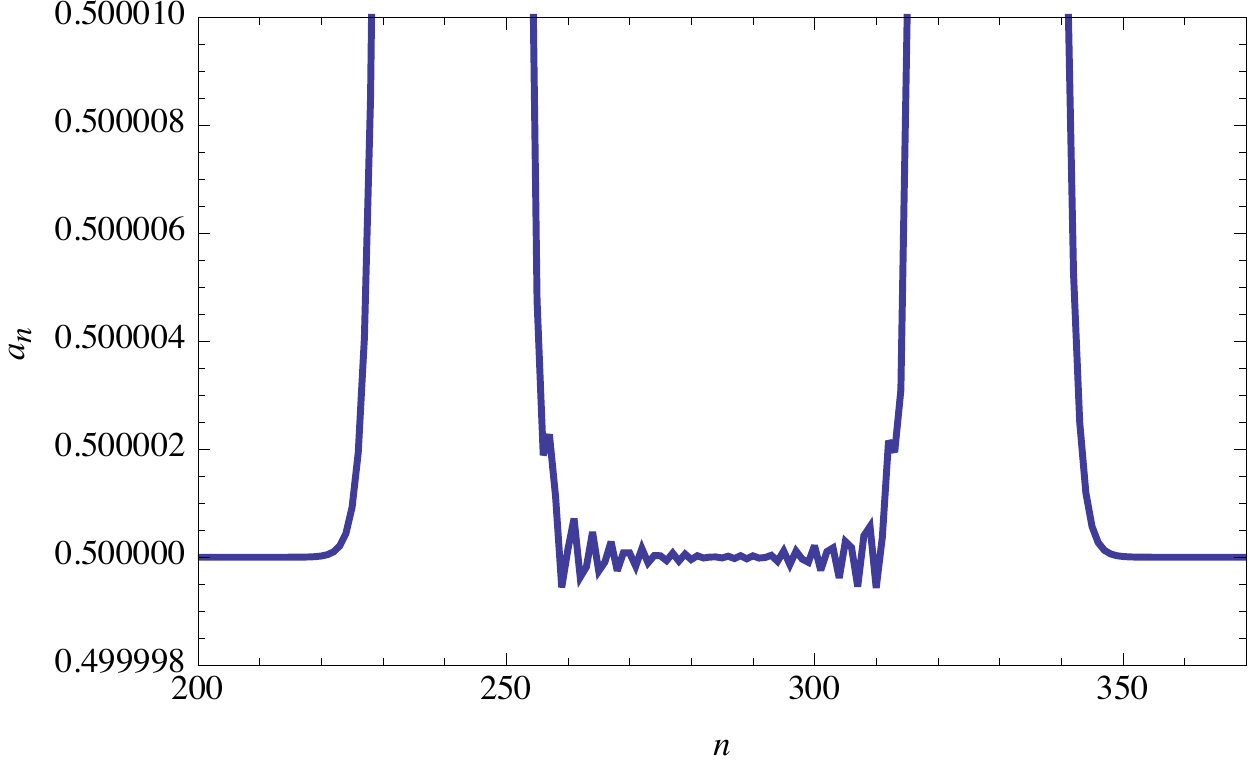}}
\caption{Head on collision of two clean solitary waves, (a) Pre-collision, (b) Post-collision.}
\label{F:SolWaveCol_closer}
\end{figure}
Now, we turn our attention to evolution of eigenvalues associated to each clean solitary wave in the numerical experiment described above. As is well known, the entire spectrum of $L$ is conserved under the Toda dynamics. Although this is not the case for the perturbed lattices, once a perturbation is fixed and a clean solitary wave solution of the perturbed system is generated, the eigenvalue associated to this clean solitary wave is found to remain asymptotically constant in time. However, the situation in case of interacting solitary waves is different. Figure~\ref{F:EVCol}(a) and (b) show time evolution of the eigenvalues corresponding to the clean solitary waves that are considered in the preceding paragraph. During the interaction, both eigenvalues deviate from their asymptotic values and move towards the edges of the continuous spectrum $[-1,1]$. As the collision state comes to an end, eigenvalues converge back to the constant asymptotic values they attain before the collision.
\begin{figure}[h!]\centering
\subfigure[Left-incident wave]{\includegraphics[width=188pt]{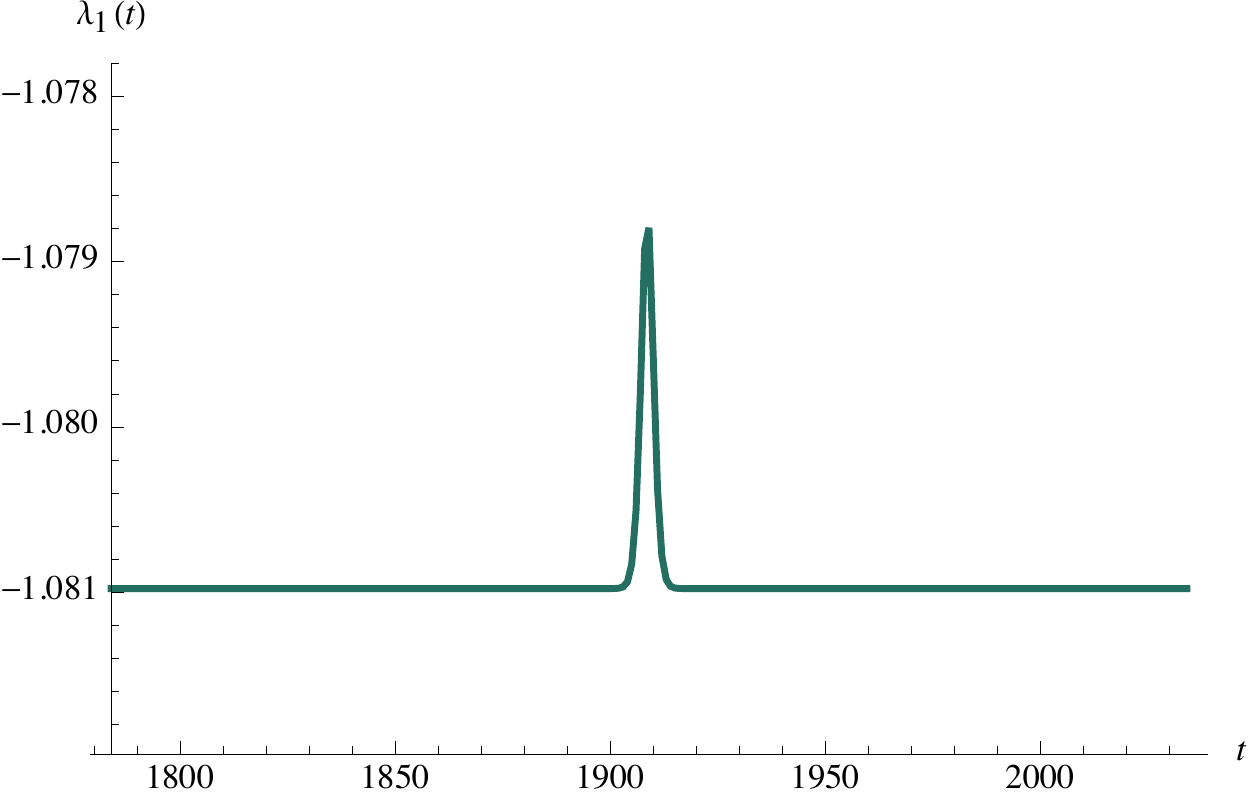}}\quad
\subfigure[Right-incident wave]{\includegraphics[width=188pt]{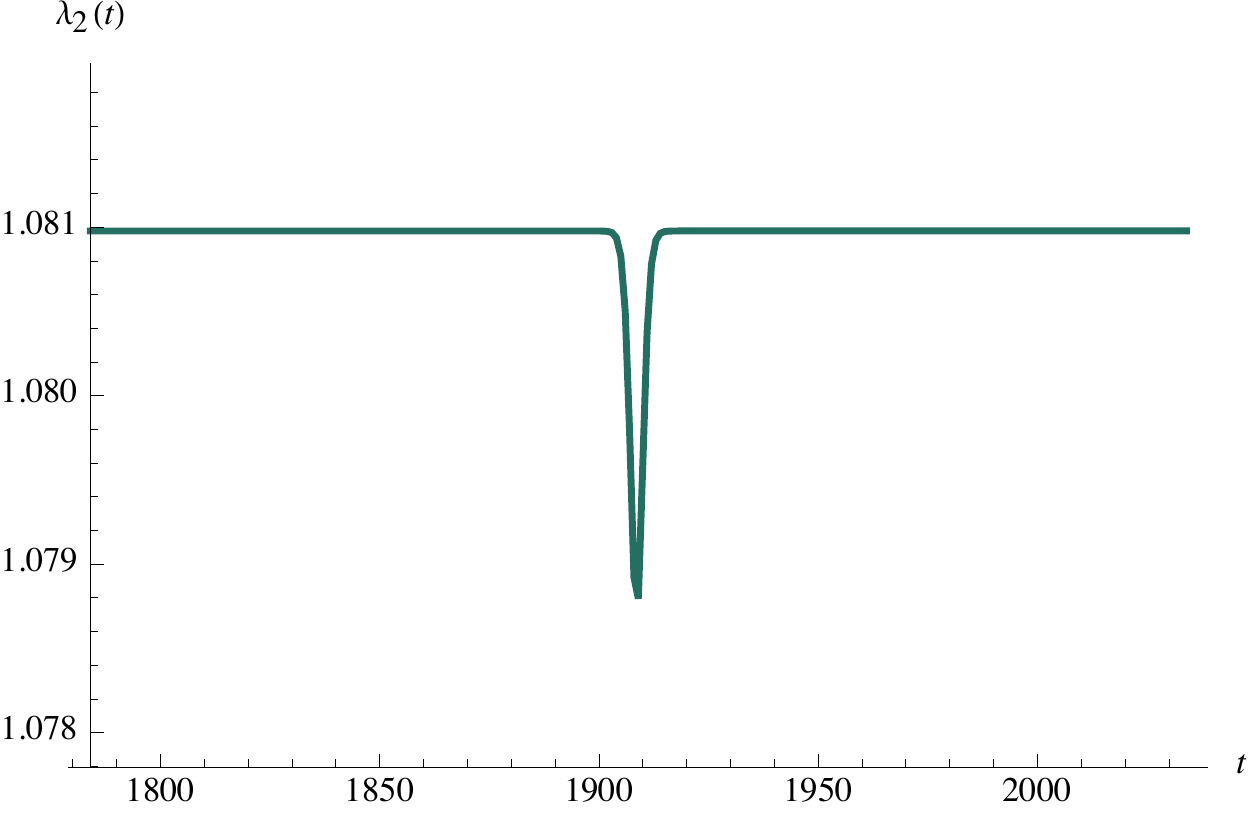}}
\caption{Eigenvalues corresponding to two clean solitary waves during their head-on collision.}
\label{F:EVCol}	
\end{figure}
In Figure~\ref{F:EVCol_3cases}, we compare evolution of eigenvalues that are associated to the solitons in the Toda lattice, to the leading solitary waves emerging from the soliton initial data in the perturbed lattice, and to the clean solitary waves generated under the same perturbation. Figure~\ref{F:EVCol_3cases}(a) displays the trajectories of the eigenvalues corresponding to the left incident traveling waves, and (b) displays those corresponding to the right incident waves. 
\begin{figure}[h!]\centering
\subfigure[Left-incident wave]{\includegraphics[width=188pt]{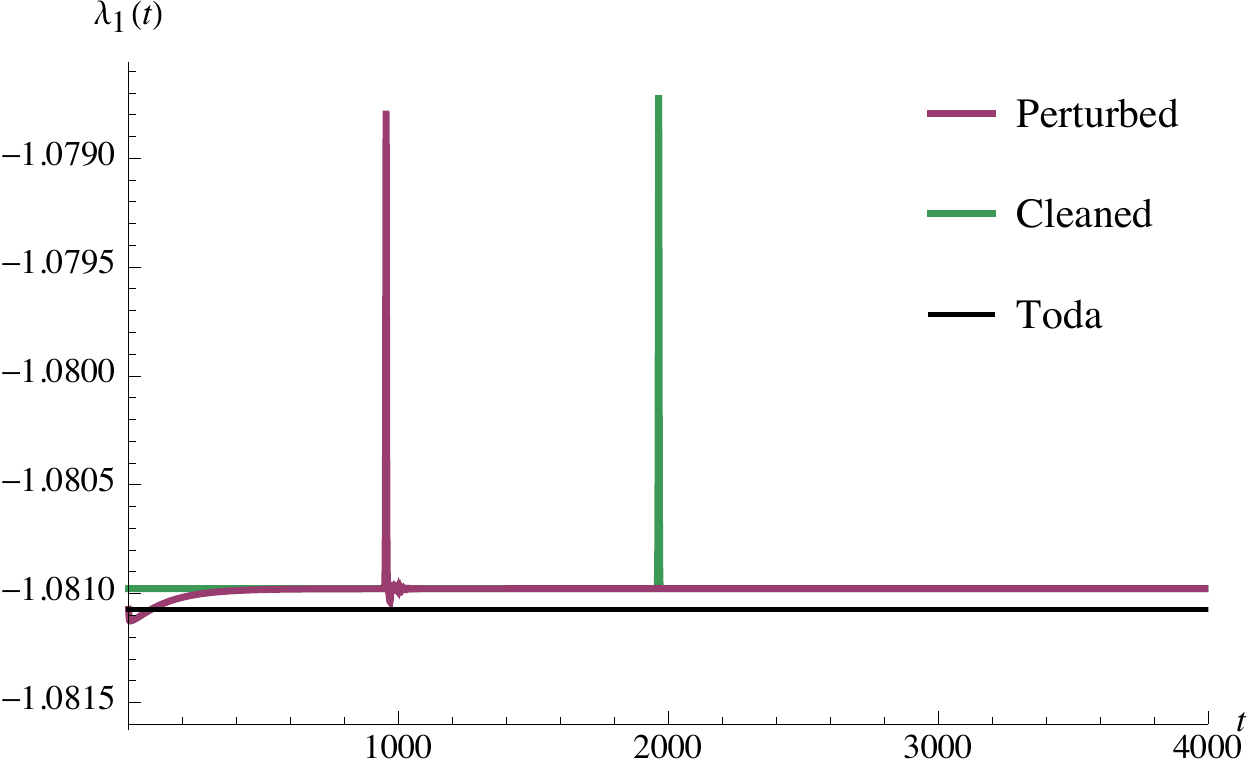}}\quad
\subfigure[Right-incident]{\includegraphics[width=188pt]{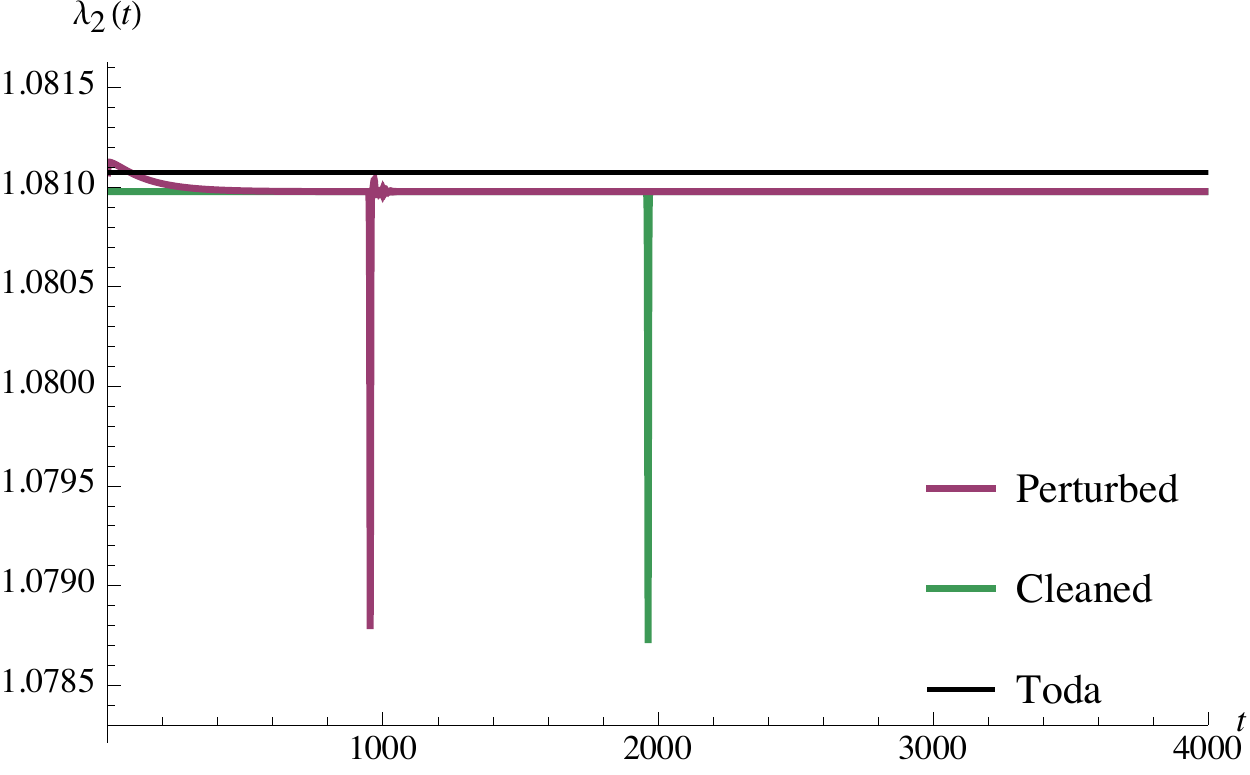}}
\caption{Eigenvalues corresponding to solitary waves during their head-on collision.}
\label{F:EVCol_3cases}
\end{figure}

\subsection{Reflection coefficient for a clean solitary wave solution}\label{S:ref}
Finally, we numerically study the time evolution of the reflection coefficient in the scattering data that is associated to a single clean solitary wave solution of the perturbed lattice. We find that the reflection coefficient associated to a clean solitary wave solution of the perturbed lattice is nontrivial; and that it exhibits oscillatory behavior, both in time and space, localized near the edges of the ac spectrum, that is, near the points $z=\pm 1$ on the unit circle. 
We numerically find that modulus of the reflection coefficient associated to a clean solitary wave solution remains asymptotically constant for all times $t\geq 0$, which is also the case in the pure Toda lattice. Figure~\ref{F:abs_refco_x2} displays the absolute value of the reflection coefficient associated to a single clean solitary wave solution of the lattice with $u(r)=r^2$ and $\varepsilon = 0.05$, at time (a) $t=0$ and (b) $t=200$.
\begin{figure}[h!]
\centering
\subfigure[$t=0$]{\includegraphics[width=397pt]{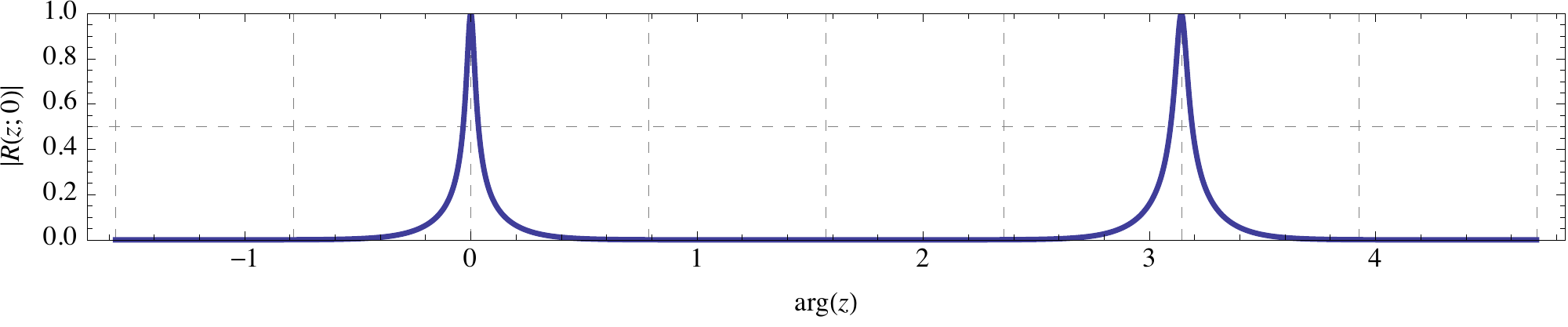}}\\
\subfigure[$t=200$]{\includegraphics[width=397pt]{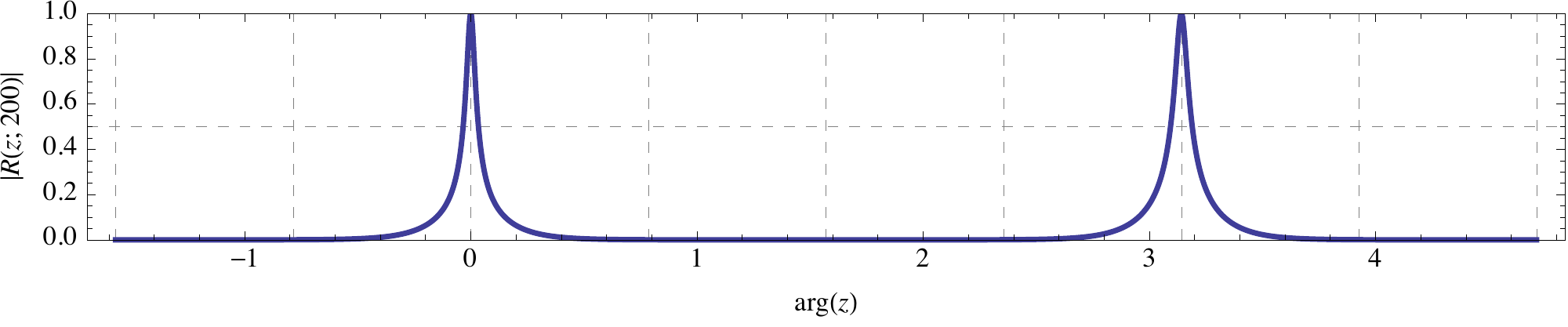}}
\caption{$|R(z;t)|$, with $u(r)=r^2$ and $\varepsilon = 0.05$; $|z|=1$, $-\pi/2 < \arg z \leq 3 \pi /2$.}
\label{F:abs_refco_x2}	
\end{figure}

\section{Discussion and conclusions} 
In this paper we have investigated the long-time behavior
of certain FPU lattices in a perturbative regime around the completely integrable Toda lattice.
We consider this question through the lens of the scattering data which can be associated,
at any time, to their solutions.  We are interested in the evolution of the scattering data since we know, 
via the bijective inverse scattering map, that each eigenvalue/norming constant pair 
corresponds to a localized ``bump" in the solution. 
If an eigenvalue is constant in time (as it is in the integrable case), 
then this ``bump" is a solitary wave traveling with constant speed.
In the perturbed FPU lattice, the evolution equations of the scattering data become much more complicated,
but we observe that the data itself continues to contain essential, readable information which describes
each part of the solution in the long time limit.

More specifically, we numerically create solitary wave solutions of the (non-integrable) FPU lattices, whose existence was proven in \cite{FW}, and identify the scattering data corresponding to such solutions by using the scattering transform for the Toda lattice. We find that eigenvalues in the scattering data remain constant for such solutions just as is the case for the integrable Toda lattice. This behavior is consistent with what was observed for a certain perturbation of the periodic Toda lattice in the relatively shorter time scales in \cite{FFM}. We also find that an eigenvalue associated with a Toda soliton initial data quickly diverges from its initial value when the soliton is let to evolve under the perturbed dynamics, but it converges to a new constant value after a short time elapses. From an inverse scattering point of view, this means that the initial poles in the associated Riemann-Hilbert problem will quickly settle at new locations in the complex plane, and then remain constant over time as is the case for the Toda lattice. We believe that these findings could serve as a stepping-stone towards performing the type of analysis as in \cite{DZ_Pert} and proving a long-time asymptotics result similar to the one in \cite{DZ_Pert} for solutions of perturbations of the doubly-infinite Toda lattice.

We also find that new eigenvalues emerge from the continuous spectrum as the solutions evolve under the perturbed dynamics. Long-time behavior of these new (and small) eigenvalues is unclear in our numerical studies, but emergence of new eigenvalues in the scattering data presents an obstacle towards using the scattering/inverse scattering transform approach employed in \cite{DZ_Pert} for the perturbed systems. We are currently investigating this phenomenon in a different project based on the numerical study of the Riemann-Hilbert problem associated with the inverse scattering transform for the Toda lattice.

\section{Acknowledgements}
\noindent The authors wish to thank Jerry Bona, Percy Deift, Fritz Gesztesy, Christian Klein, Ken McLaughlin, Peter Miller, and Tom Trogdon for useful discussions and suggestions. Both authors acknowledge the support of the National Science Foundation through NSF grants DMS-1150427 and DMS-0845760.

\appendix
\section{Proofs and remarks}
\label{A:app}
In this appendix, we collect some of the longer proofs of the theorems from the first part of the paper. We include these proofs
for completeness, since the methods are in most cases fairly standard, but still needing to be adapted to our case of the perturbation of the Toda lattice.

More can be shown regarding the solutions $(a, b)$ of the perturbed lattice than we discussed in Section~\ref{S:Background}. 
We let $\ell^1_w(\mathbb{Z})$ denote the weighted $\ell^1$-space, with the weight function given by \mbox{$n \mapsto 1 + |n|$}, 
and define the Banach space $X_w^1  = \ell^{1}_w(\mathbb{Z}) \oplus \ell^{1}_w(\mathbb{Z})$ equipped with the norm
\begin{equation*}
\big\| (x ,y) \big\|_{w,1} = \sum_{n\in\mathbb{Z}}\big(1 + |n|\big)\big( | x_n | + | y_n | \big)\,.
\end{equation*}
Then we have the following result, which establishes further spatial decay (needed in order to set up the scattering theory) 
for the solutions of the perturbed lattice.
\begin{theorem}\label{T:SpatialDecay}
Let $V_\varepsilon$ be a potential satisfying Assumption~\ref{A}, and $\tilde{a}^0$ and $b^0$ be bounded sequences such that $\left(\tilde{a}^0 , b^0\right)\in X_\text{w}^1$.
Let $\big(a(t), b(t)\big)$ be the unique global solution of the perturbed lattice \eqref{E:eom_p_ab} 
corresponding to the initial conditions
\begin{equation*}
a(0) = \tfrac{1}{2} + \tilde{a}^0 > 0~\text{ and }~b(0) = b^0\,.
\end{equation*}
Then
$
\Big(a(t) - \tfrac{1}{2}\,, \,b(t) \Big) \in X_w^1\,,
$
for all times $t \geq  0$.
\end{theorem}

In the case of the Toda lattice ($\varepsilon=0$), this result was proven by G. Teschl \cite{Tesch02}. Our proof
is a modification of his, adapted to deal with the perturbation term in the evolution equation.

\begin{proof}
We consider $\tilde{a}(t) = a(t) - \frac{1}{2}$\,, with $\tilde{a}(0) = \tilde{a}^0$\,, and study the differential equation which governs the evolution of $(\tilde{a}, b)$:
\begin{equation}\label{E:spatial_ode}
\partial_t 
\begin{pmatrix} \tilde{a}_n (t)\\ b_n (t)\end{pmatrix} =
\begin{pmatrix}
a_n(t) \big( b_{n+1} (t) - b_n (t) \big)\\
\big(2a_n(t) + 1 \big) \tilde{a}_n(t) - \big(2a_{n-1}(t) + 1 \big) \tilde{a}_{n-1}(t)+ \varepsilon \tilde{U}_{nn}(t)
\end{pmatrix}
\end{equation}
subject to initial conditions $\big( \tilde{a}(0), b(0 )\big) = \left(\tilde{a}^0, b^0\right) \in X^{1}_w$, where
\begin{equation*}
\tilde{U}_{nn}(t) = \frac{1}{2}\Big\{ u' \Big(-2\log{\big(2\tilde{a}_{n-1}(t) + 1\big)}\Big) - u' \Big(-2\log{\big(2\tilde{a}_{n}(t) + 1\big)}\Big) \Big\}\,,
\end{equation*}
for each $n$ in $\mathbb{Z}$. Let $f\big(t,\big(\tilde{a}(t), b(t)\big)\big)$ denote the right hand side of \eqref{E:spatial_ode}. By our assumptions on the initial data, we can choose
\begin{equation*}
\delta = \tfrac{1}{2}\min_{n\in\mathbb{Z}} \Big( \tilde{a}_n^{0} + \tfrac{1}{2}\Big)
\end{equation*} 
so that for any $(x,y)$ in the ball $B_{\delta} = \left\lbrace (x,y)  \colon \big\|(x,y)-\big(\tilde{a}^{0}, b^{0}\big)\big\|_{w,1} \leq \delta \right\rbrace$, $B_\delta \in X^1_w$, we have $\min_{n\in\mathbb{Z}} x_n >0$. Now, since the weight function $n \mapsto 1 +  |n|$ satisfies
\begin{equation}
\sup_{n\in\mathbb{Z}} \left\lbrace \tfrac{1+|n+1|}{1+|n|},  \tfrac{1+|n|}{1+|n+1|}\right\rbrace < \infty\,,
\end{equation}
the shift operators are bounded with respect to the norm $\| \cdot \|_{w,1}$. The multiplication operator with the sequence $a(t)$ is also uniformly bounded from $X^1_w$ into $X^1_w$ as $\|a(t)\|_{\ell^\infty}$ is bounded uniformly in time. Moreover, since $u\in C^2(\mathbb{R})$ and $r \mapsto \log r$ is Lipschitz continuous on $[\rho, +\infty)$ for any $\rho >0$, the map
\begin{equation*}
\begin{pmatrix} x_n\\ y_n \end{pmatrix}_{n\in\mathbb{Z}} \mapsto  \frac{1}{2}\begin{pmatrix} 0 \\
u' \Big(-2\log{\big(2x_{n-1}+ 1\big)}\Big) - u' \Big(-2\log{\big(2x_{n} + 1\big)}\Big)
\end{pmatrix}_{n\in\mathbb{Z}}
\end{equation*}
is Lipschitz continuous on $B_\delta \subset X^1_w$ by our choice of $\delta$. Therefore, there exists $T>0$ such that the map
\begin{equation*}
\begin{pmatrix} x(t)\\ y(t) \end{pmatrix} \mapsto  \begin{pmatrix}\tilde{a}^0 \\ b^0 \end{pmatrix} +
\int_0^t{f\Big(\tau, \big(x(\tau),y(\tau)\big)\Big)\,d\tau}
\end{equation*}
defined on the Banach space $X_T = [0,T] \times X_{w,1}$ equipped with the norm
\begin{equation*}
\| \cdot \|_{X_T} = \sup_{t\in [0,T]}\| \cdot \|_{w,1}\,,
\end{equation*}
is a contraction mapping from the closed ball $B_{\delta, T} = [0,T] \times B_{\delta} \subset X_T$ into itself, and the Banach 
Fixed Point Theorem implies that \eqref{E:spatial_ode} has a unique solution $\big(\tilde{a}(t), b(t)\big) \in X^1_w$ for $t\in[0,T]$ 
with $\big(\tilde{a}(0), b(0)\big) = \big(\tilde{a}^0, b^0\big)$. A straightforward calculation shows that this solution depends 
continuously on the initial data and $\big(\tilde{a}(t) + \tfrac{1}{2}\,, \,b(t) \big)$ coincides with the solution $\big(a(t)\,, \,b(t) \big)$ 
of the perturbed lattice. Now, suppose $(0,t^{*})$ is some finite time interval for which the solution to \eqref{E:spatial_ode} exists. 
Since the solution satisfies $(a_n)_{n\in\mathbb{Z}},(1/a_n)_{n\in\mathbb Z},(b_n)_{n\in\mathbb{Z}}
\in\ell^\infty(\mathbb{Z})$, we have 
\begin{equation*}
\big\| \big(0, \tilde{U}_{nn}(t)\big)_{n\in\mathbb{Z}}\big\|_{w,1} \leq \varepsilon 6 K \big\| \big(\tilde{a}(t), b(t) \big) \big\|_{w,1}\,
\end{equation*}
for any $t\in(0,t^*)$, where $K$ is the product of the Lipschitz constants of $u$ and $r \mapsto \log{r}$, which depend only on the initial data.
Then
\begin{equation*}
\big\| \big(\tilde{a}(t), b(t) \big)\big\|_{w,1}  \leq \big\| \big(\tilde{a}^0, b^0 \big)\big\|_{w,1} + \int_0^{t}{\big(6(C+\varepsilon K) + 3\big)\big\| \big(\tilde{a}(\tau), b(\tau) \big)\big\|_{w,1}    \,d\tau}\,
\end{equation*}
for any time $t\in(0, t^*)$, where $C = \sup_{t\geq 0} \| a(t)\|_{\ell^{\infty}}$. Using Gr\"{o}nwall's inequality, we obtain
\begin{equation*}
\big\| \big(\tilde{a}(t), b(t) \big)\big\|_{w,1}  \leq \big\| \big(\tilde{a}^0, b^0 \big)\big\|_{w,1} e^{(6(C+\varepsilon K) + 3)t}\,.
\end{equation*}
Therefore the solutions of \eqref{E:spatial_ode} remain bounded on finite time intervals and hence they are global in time in $X^1_w$.
\end{proof}

In the second part of the appendix we include the  evolution equations for the scattering data associated to
the Jacobi matrices that we are studying.
Note that the scattering relations \eqref{E:scatteringsoln1} and \eqref{E:scatteringsoln2} can be rewritten as
\begin{equation}\label{scat_rels}
\big(\varphi_+(z,n)\,\,\,\varphi_+(z^{-1},n)\big)
=\big(\varphi_-(z^{-1},n)\,\,\,\varphi_-(z,n)\big)\cdot \bd S(z)\,,
\end{equation}
where $\bd S(z)$ is the $2\times2$ scattering matrix. It can easily be related to the transmission
and reflection coefficients:
\begin{equation}\label{E:scat_matrix}
\bd S(z)=\frac{1}{T(z)}\cdot
\left(\begin{matrix}
1 & -R_+(z)\\
R_-(z) & T(z)^2-R_+(z)R_-(z)
\end{matrix}\right)\,.
\end{equation}

Now we take into account the time dependence. The time evolution of the reflection coefficient under the 
Toda lattice is given by
$R(z,t) = R_0(z)e^{t(z-z^{-1})}$,
where $R_0(z)$ is the reflection coefficient of the initial Jacobi matrix $L(t=0)$. 
Under the perturbed lattice evolution, however, additional terms appear. We begin with
a technical, but very important, result:
\begin{theorem}\label{T:eom_p_S}
The evolution equation induced by \eqref{E:eom_p_L} on the scattering matrix $\bd S$ defined
in \eqref{scat_rels} and \eqref{E:scat_matrix} is given, for all $|z|=1$, $z\neq \pm1$, and all $t\in\mathbb{R}$, by
\begin{equation}\label{E:Evol_p_S}
\partial_t \bd{S}(z,t) + \theta(z)[\bd{S}(z,t),\sigma_3] = \frac{\varepsilon}{\theta(z)}\bd{D}(z,t)\,,
\end{equation}
where
\begin{equation}\label{E:DefnTheta}
\theta(z)=\frac{z-z^{-1}}{2}\,,
\end{equation}
and
\begin{equation*}
\bd{D}(z,t) = \sum_{n=-\infty}^{\infty}U_{nn}(t)
\begin{pmatrix}
\varphi_+(z;n,t)\varphi_-(z;n,t) & \varphi_+(z^{-1};n,t)\varphi_-(z;n,t)\\
-\varphi_+(z;n,t)\varphi_-(z^{-1};n,t) & -\varphi_+(z^{-1};n,t)\varphi_-(z^{-1};n,t)
\end{pmatrix}.
\end{equation*}
\end{theorem}

Rewriting \eqref{E:Evol_p_S} as
\begin{equation*}\label{E:4:49}
\partial_t \bd{S} = -\theta
\left(\begin{matrix} 0 & -2S_{12} \\ 2S_{21} & 0 \end{matrix}\right)
+\frac{\varepsilon}{\theta}\bd{D}\,,
\end{equation*}
together with the equation \eqref{E:scat_matrix} for the reflection coefficient in terms of the entries of $\bd{S}$
leads immediately to the evolution equation for $R$:
\begin{equation}\label{E:eom_p_R_+}
\partial_t R(z;t) = 2 \theta(z) R(z;t) -\frac{\varepsilon\, T(z;t)^2}{\theta(z)}\sum_{n=-\infty}^{\infty} U_{nn}(t)\varphi_-(z;n,t)^2\,,
\end{equation}
for all $t\in\mathbb{R}$ and all $|z|=1$, $z\neq\pm1$, where $\theta(z)$ was defined in \eqref{E:DefnTheta}.

The proof of this result is based on a discrete version of the variation of constants technique.
While in this case the adaptation of the continuous technique to the discrete setting requires some
well-made choices along the way, it is still fairly standard, and similar, for example, to the derivation in \cite{DZ_Pert} of the 
analogous equation for the reflection coefficient
in a perturbed defocusing cubic nonlinear Schr\"odinger equation.

\begin{proof}[Proof of Theorem~\ref{T:eom_p_S}]
In order to streamline our notation, we consider for each $z$ with $|z|=1$ two ``$\infty\times 2$'' matrices:
\begin{equation*}
\Psi(z;t)=\big(\Psi_n(z;t)\big)_{n\in\mathbb{Z}}=\big(\psi_{n,j}(z;t) \big)_{\substack{n\in\mathbb{Z} \\  j=1,2}}
\text{ and }
\Phi(z;t)=\big(\Phi_n(z;t)\big)_{n\in\mathbb{Z}}=\big(\phi_{n,j}(z;t) \big)_{\substack{n\in\mathbb{Z} \\  j=1,2}}\,,
\end{equation*}
with
$\psi_{n,1}(z;t) = \varphi_+(z;n,t)$, $\psi_{n,2}(z;t) = \varphi_+(z^{-1};n,t)$,
$\phi_{n,1}(z;t) = \varphi_-(z^{-1};n,t)$, and $\phi_{n,2}(z;t) = \varphi_-(z;n,t)$.
Furthermore, in the proofs and some of the statements in this section, we will avail ourselves
of the common convention of suppressing the dependence on the $z$ and/or $t$ variables,
in order to keep the length of some of the formulas manageable.

We begin by differentiating both sides of the equation
$L(t)\Psi(z;t) = \frac{z+z^{-1}}{2}\Psi(z;t)$
with respect to $t$. Using the evolution equation \eqref{E:eom_p_L} for $L$ we obtain:
\begin{equation}\label{E:4:5}
\left( L - \frac{z+z^{-1}}{2} \right)\big( \partial_t\Psi(z) - P \Psi(z) \big) = -\varepsilon U \Psi(z)\,.
\end{equation}
For each $n\in\mathbb{Z}$, we seek $\bd{A}_n=\bd{A}_n(z;t)$ a $2 \times 2$ matrix such that
\begin{equation}\label{E:4:6}
\partial_t\Psi - P \Psi = \Psi
\left(\begin{matrix}-\theta &0 \\
0 & \theta \end{matrix}\right)
+ \big(\Psi_n\bd{A}_n\big)_{n\in\mathbb{Z}}\,,
\end{equation}
where $\theta=\theta(z)$ as in \eqref{E:DefnTheta}. We note that the first term on the right hand side is the one 
obtained in the usual Toda case, while the second term encodes the variation of constants. 
Plugging \eqref{E:4:6} in \eqref{E:4:5} 
followed by straightforward, if quite lengthy, calculations allow us to simplify the left-hand side, leading to
\begin{equation}\label{E:4:18}
a_{n-1}\Psi_{n-1}\big(\bd{A}_{n-1} - \bd{A}_n\big)
- a_{n}\Psi_{n+1}\big(\bd{A}_n - \bd{A}_{n+1}\big) = -\varepsilon U_{nn}\Psi_n\,.
\end{equation}
We seek $\big(\bd{A}_n\big)_n$ such that
$\Psi_n\big( \bd{A}_{n+1} - \bd{A}_n \big) = 0$ for all $n\in\mathbb{Z}$.
Using this assumption in \eqref{E:4:18} and combining everything
yields a single matrix equation:
\begin{equation*}
\begin{pmatrix}
-a_n\Psi_{n+1} \\
a_n\Psi_n
\end{pmatrix}
\big(\bd{A}_n - \bd{A}_{n+1} \big) =
\begin{pmatrix}
-\varepsilon U_{nn}\Psi_n \\
0
\end{pmatrix}\,.
\end{equation*}
Direct calculations, as well as standard arguments involving the Wronskian of the (linearly independent)
Jost solutions $\varphi_+(z,\cdot)$ and $\varphi_+(z^{-1},\cdot)$, show that
$\det \begin{pmatrix} -a_n\Psi_{n+1} \\ a_n\Psi_n \end{pmatrix} =a_n\frac{z^{-1} - z }{2}= -a_n \theta \neq 0\,.$
Hence we can invert the $2\times2$ matrix $\left(\begin{matrix} -a_n \Psi_{n+1} \\ a_n\Psi_n \end{matrix}\right)$
and sum the results to obtain
\begin{equation}\label{E:4:33}
\bd{A}_n=\frac{\varepsilon}{\theta}\sum_{j=n}^{\infty}U_{jj}
\left( \begin{matrix}
\psi_{j,1}\psi_{j,2} & \psi_{j,2}^2 \\
-\psi_{j,1}^2 & -\psi_{j,1} \psi_{j,2}
\end{matrix}\right)\,.
\end{equation}

We will now obtain the evolution equation for the scattering matrix $\bd S$.
We start with rewriting \eqref{E:4:6} as
\begin{equation*}\label{E:4:34}
\partial_t\Psi - P \Psi + \theta\Psi\sigma_3 = \big(\Psi_n\bd{A}_{n}\big)_{n\in\mathbb{Z}}\,,
\end{equation*}
and using \eqref{scat_rels} in the form: 
\begin{equation}\label{E:scat_rels2}
\Psi_n = \Phi_n \bd{S}\qquad\text{for all}\,\,\, n\in\mathbb{Z}\,.
\end{equation}
For each $n\in\mathbb{Z}$, this leads to the equation: 
\begin{equation}\label{E:4:36}
\left(\partial_t \Phi_n \right) \bd{S} + \Phi_n \left( \partial_t \bd{S}\right) - \big( P \Phi\big)_n \bd{S} 
+\theta\Phi_n \bd{S} \sigma_3 = \Phi_n \bd{S} \bd{A}_n\,.
\end{equation}
From the expression \eqref{E:4:33} for $\bd A_n$ and from \eqref{E:scat_rels2} used for
all $j\geq n$, we obtain that 
\begin{equation*}
\bd{S} \bd{A}_n = \frac{\varepsilon}{\theta}\sum_{j=n}^{\infty}U_{jj}
\begin{pmatrix}
\psi_{j,1}\phi_{j,2} & \psi_{j,2}\phi_{j,2}\\
-\psi_{j,1}\phi_{j,1} & -\psi_{j,2}\phi_{j,1}
\end{pmatrix}
\end{equation*}
Taking $n\to-\infty$ in \eqref{E:4:36} and using the
asymptotic properties of the Jost solutions, we see that
\begin{equation*}
\lim_{n\to-\infty} \Phi_n(z)\cdot z^{-n\sigma_3}=(1\,\,\,1)\quad
\text{where }
\sigma_3=
\left(
\begin{matrix}1&0\\0&-1\end{matrix}\right)\text{ and }
z^{\pm n\sigma_3}=\left(
\begin{matrix}z^{\pm n}&0\\0&z^{\mp n}\end{matrix}\right)\,.
\end{equation*}
Further standard scattering theory arguments show that, as $n\to-\infty$,
$\partial_t\Phi_n\to (0\,\,\,0)$
and\\
$\big(P\Phi \big)_n \cdot z^{-n\sigma_3} \to\theta(z)(1\,\,\,-1)$.
Using these asymptotics as $n \rightarrow - \infty$ in \eqref{E:4:36} then yields: 
\begin{equation*}\label{E:4:47}
\partial_t \bd{S} + \theta[\bd{S},\sigma_3] = \frac{\varepsilon}{\theta}\bd{D}\,,
\end{equation*}
which completes the proof.
\end{proof}

Under the Toda evolution, $\zeta_k(t)=\zeta_k(0)$ and
$\gamma_k (t) = \gamma_k(0) e^{t(\zeta_k(0)-\zeta_k(0)^{-1})}$. 
Under the perturbed lattice, the evolution of the eigenvalues and norming constants is given by:
\begin{theorem}\label{T:eom_p_eigs}
The evolution equations for the eigenvalues and norming constants under the perturbed lattice 
are given by:
\begin{equation}\label{E:eom_p_eigenvalues}
\partial_t \lambda_k(t)=\varepsilon\gamma_k(t)\sum_{n=-\infty}^\infty U_{nn}(t)\varphi_+\big(\zeta_k(t);n,t\big)^2\,,
\end{equation}
and
\begin{equation}\label{E:eom_p_norm}
\partial_t \gamma_k(t)=2\theta(\zeta_k(t))\gamma_k(t) 
- \frac{2\varepsilon\gamma_k^2(t)}{\theta\big(\zeta_k(t)\big)} \sum_{n=-\infty}^\infty\varphi_+(\zeta_k(t);n,t)
\times 
\sum_{j=n}^{\infty}U_{jj}(t)\varphi_+(\zeta_k(t);j)K(j,n;t)\,,
\end{equation}
where
\begin{equation*}\label{E:2:26}
K(j,n;t)= \varphi_+(\zeta_k(t);n,t)\varphi_+(\zeta_k(t)^{-1};j,t)- \varphi_+(\zeta_k(t)^{-1};n,t)\varphi_+(\zeta_k(t);j,t)\,.
\end{equation*}
\end{theorem}

Note that we cannot a-priori say that equations \eqref{E:eom_p_eigenvalues} and \eqref{E:eom_p_norm} hold
for all time $t>0$ (unlike \eqref{E:eom_p_R_+}). Rather the equations hold locally in time, assuming that
we start at an eigenvalue. Recall, however, that the eigenvalues of our Jacobi matrices are all simple:
in particular, this means that eigenvalues cannot cross, but they can stop existing by ``entering" the ac spectrum. 
We do not observe this phenomenon in any of our numerical simulations, but it remains a theoretical possibility. 
The proof of  \eqref{E:eom_p_eigenvalues} is standard,
and arguments very similar to those of the proof of Theorem~\ref{T:eom_p_S} lead to \eqref{E:eom_p_norm}.

\nocite{*}
\bibliographystyle{plain} 
\bibliography{NumericsFPU}



\end{document}